\documentclass[a4paper,11pt,centertags]{amsart}
\pdfoutput=1
\usepackage{amsfonts,amscd,amsmath,amssymb,amsthm}
\usepackage{times}
\usepackage{graphicx}

\makeatletter
\newif\if@golden  \@goldentrue
\newcommand{\f@ctor}{1}
\newlength{\aiv@width}  
\newlength{\aiv@height} 
\newlength{\tmp@width}  
\newlength{\tmp@height} 
\if@twocolumn\if@golden
  \else\fi
\else\if@golden\ifcase\@ptsize\relax\or
\or\fi
  \else\ifcase\@ptsize\relax\or
\or\fi\fi
\if@twocolumn\if@golden\ifcase\@ptsize\relax\renewcommand{\f@ctor}{53}
  \or\renewcommand{\f@ctor}{46}\or\renewcommand{\f@ctor}{43}\fi
  \else\ifcase\@ptsize\relax\renewcommand{\f@ctor}{51}\or
  \renewcommand{\f@ctor}{45}\or\renewcommand{\f@ctor}{42}\fi\fi
\else\if@golden\ifcase\@ptsize\relax \renewcommand{\f@ctor}{46}
  \or\renewcommand{\f@ctor}{43}\or\renewcommand{\f@ctor}{43}\fi
  \else\ifcase\@ptsize\relax\renewcommand{\f@ctor}{43}
  \or\renewcommand{\f@ctor}{40}\or\renewcommand{\f@ctor}{40}\fi\fi\fi
\multiply\textheight by \f@ctor
\makeatother

\numberwithin{equation}{section}

\DeclareMathAlphabet{\mathsf}{OT1}{phv}{m}{n}
\DeclareMathAlphabet{\mathrm}{OT1}{ptm}{m}{n}

\DeclareSymbolFont{ER}{U}{eur}{m}{n}
\DeclareSymbolFont{SY}{U}{psy}{m}{n}

\DeclareMathSymbol{,}{\mathpunct}{SY}{'054}
\DeclareMathSymbol{.}{\mathpunct}{SY}{'056}
\DeclareMathSymbol{:}{\mathpunct}{SY}{'072}

\DeclareMathSymbol{(}{\mathopen}{SY}{'050}
\DeclareMathSymbol{)}{\mathclose}{SY}{'051}

\DeclareMathSymbol{+}{\mathbin}{SY}{'053}
\DeclareMathSymbol{-}{\mathbin}{SY}{'055}
\DeclareMathSymbol{=}{\mathbin}{SY}{'075}

\DeclareMathSymbol{<}{\mathbin}{SY}{'074}
\DeclareMathSymbol{>}{\mathbin}{SY}{'076}
\DeclareMathSymbol{\leq}{\mathbin}{SY}{'243}
\DeclareMathSymbol{\geq}{\mathbin}{SY}{'263}
\DeclareMathSymbol{\nneq}{\mathbin}{SY}{'271}
\DeclareMathSymbol{\in}{\mathbin}{SY}{'316}
\DeclareMathSymbol{\nnotin}{\mathbin}{SY}{'317}
\DeclareMathSymbol{\times}{\mathbin}{SY}{'264}
\DeclareMathSymbol{\pm}{\mathbin}{SY}{'261}
\DeclareMathSymbol{\subset}{\mathbin}{SY}{'314}
\DeclareMathSymbol{\supset}{\mathbin}{SY}{'311}
\DeclareMathSymbol{\subseteq}{\mathbin}{SY}{'315}
\DeclareMathSymbol{\supseteq}{\mathbin}{SY}{'312}
\DeclareMathSymbol{/}{\mathord}{SY}{'057}
\DeclareMathSymbol{\ast}{\mathord}{SY}{'052}
\DeclareMathSymbol{\perp}{\mathord}{SY}{'136}

\renewcommand{\neq}{\nneq}

\renewcommand{\notin}{\nnotin}

\renewcommand{\theequation}{\arabic{section}.\arabic{equation}}

\newcommand{\Z}{\mathbb{Z}}

\newcommand{\R}{\mathbb{R}}
\newcommand{\C}{\mathbb{C}}
\newcommand{\N}{\mathbb{N}}

\newcommand{\cC}{\mathcal{C}}
\newcommand{\cD}{\mathcal{D}}
\newcommand{\cE}{\mathcal{E}}
\newcommand{\cF}{\mathcal{F}}
\newcommand{\cG}{\mathcal{G}}
\newcommand{\cH}{\mathcal{H}}
\newcommand{\cI}{\mathcal{I}}

\newcommand{\cK}{\mathcal{K}}
\newcommand{\cL}{\mathcal{L}}
\newcommand{\cM}{\mathcal{M}}

\newcommand{\cP}{\mathcal{P}}

\newcommand{\cT}{\mathcal{T}}

\newcommand{\cV}{\mathcal{V}}

\newcommand{\dist}{{\ensuremath{\mathrm{dist}}}}

\DeclareMathOperator{\Ker}{\mathrm{Ker}}

\newcommand{\1}{\mathbb{I}}
\newcommand{\spec}{{\ensuremath{\rm spec}}}

\DeclareMathSymbol{\emptyset}{\mathord}{SY}{'306}
\DeclareMathSymbol{\oplus}{\mathord}{SY}{'305}

\newtheorem{theorem}{Theorem}{\bf}{\it}
\newtheorem{proposition}[theorem]{Proposition}{\bf}{\it}
\newtheorem{corollary}[theorem]{Corollary}{\bf}{\it}
{\bf}{\it}
{\it}{\rm}
\newtheorem{lemma}[theorem]{Lemma}{\bf}{\it}
{\it}{\rm}
\newtheorem{definition}[theorem]{Definition}{\bf}{\it}
{\bf}{\it}
{\bf}{\it}

\newcommand{\ga}{\ensuremath \alpha}
\newcommand{\aitem}{\renewcommand{\labelenumi}{(\alph{enumi})}}

\title[Finite propagation speed for solutions of the wave equation on metric graphs]
{Finite propagation speed for solutions of the wave equation on metric graphs}

\author[V. Kostrykin]{Vadim Kostrykin}
\address{Vadim Kostrykin\\Institut f\"{u}r Mathematik, Johannes Gutenberg- Universit\"at,
D-55099 Mainz, Germany}
\email{kostrykin@mathematik.uni-mainz.de}

\author[J. Potthoff]{J\"urgen Potthoff}
\address{J\"urgen Potthoff\\ Institut f\"ur Mathematik, Universit\"at Mannheim,
D-68131 Mannheim, Germany}
\email{potthoff@math.uni-mannheim.de}

\author[R. Schrader]{Robert Schrader}
\address{Robert Schrader\\ Institut f\"{u}r
Theoretische Physik\\ Freie Universit\"{a}t Berlin, Arnimallee 14
\\ D-14195 Berlin,
Germany}
\email{robert.schrader@fu-berlin.de}

\keywords{Metric graphs, Laplace operators, wave equation, finite propagation speed}

\subjclass[2010]{34B45,35L05,35L20}

\date{June~4, 2011}

\begin{document}

\begin{abstract}
We provide a class of self-adjoint Laplace operators $-\Delta$ on metric graphs with
the property that the solutions of the associated wave equation satisfy the finite
propagation speed property. The proof uses energy methods, which are adaptions of
corresponding methods for smooth manifolds.
\end{abstract}

\maketitle

\section{Introduction}
Nature tells us that energy and information can only be transmitted with finite
speed, smaller or equal to the speed of light. The mathematical framework, which
allows an analysis and proof of this phenomenon, is the theory of hyperbolic
differential equations and in particular of the wave equation
\begin{equation*}\label{waveeq}
    \Box \psi=0
\end{equation*}
where  $\Box\; =\;\partial_t^2-\Delta$ is the d'Alembert operator with $-\Delta$ as
the Laplace operator, and $t\in\R$ is a time parameter. The result, which may be
obtained, runs under the name \emph{finite propagation speed}. The configuration space
and hence the context, within which the wave operator and finite propagation speed can
be discussed, may be an arbitrary manifold in which the notions both of a distance
between two points and of a Laplace operator makes sense. In more detail, given the
Laplacian $-\Delta$ and hence the associated d'Alembert operator, the central
quantity entering the construction and discussion of solutions of the wave equation
for given Cauchy data (initial conditions) is the wave kernel
\begin{equation*}
    W(t)=\frac{\sin\big(\sqrt{-\Delta}t\bigr)}{\sqrt{-\Delta}},\qquad t\in\R.
\end{equation*}
Let $W(t)(p,q)$ denote the associated integral kernel. Then finite propagation speed
is a general result on hyperbolic equations and the statement that $W(t)(p,q)$
vanishes whenever $|t|<{\rm distance}(p,q)$. For an extensive text book discussion,
see e.g.~\cite{Evans, Taylor, TaylorI}.

The d'Alembert operator and the associated Klein-Gordon operator $\Box+m^2$ play an
important r\^ole in relativistic quantum theories, see e.g.\ standard text books on
relativistic quantum field theory like \cite{Itzykson_Zuber, Schweber, Weinberg}.
Free quantum fields of mass $m>0$ satisfy the Klein-Gordon equation. Thus a quantum
version of finite propagation speed is the condition that space-like separated
observables commute. Since the fundamental article of Wightman \cite{Wightman},
this condition is considered as indispensable for any local relativistic quantum
theory \cite{Haag, Jost, Streater_Wightman}. Thus, the commutator of a
hermitean, free, massive,
scalar boson field $\Phi(x,t)$ on Minkowski space $\cM=\R^4$ is
given by the integral kernel associated to the wave kernel
\begin{equation*}
    W_m(t)=\frac{\sin\bigl(\sqrt{-\Delta+m^2}t\bigr)}{\sqrt{-\Delta+m^2}}
\end{equation*}
of the Klein-Gordon operator, that is
\begin{equation*}\label{qfinprop}
    \left[\Phi(\vec{x},t),\Phi(\vec{y},s)\right]= i\, W_m(t-s)(\vec{x},\vec{y})
\;\1 \qquad
    (\vec{x},t),\;(\vec{y},s)\in\cM.
\end{equation*}
Two events $(\vec{x},t)$ and $(\vec{y},s)$ are space-like separated if
$|\vec{x}-\vec{y}|^2>(t-s)^2$, in units, where the speed of light equals $1$. Thus local
commutativity in this context is the property
\begin{equation*}\label{qfinprop-i}
    W_m(t-s)(\vec{x},\vec{y})=0, \quad
        \mbox{\it if the events $(\vec{x},t)$ and $(\vec{y},s)$ are space-like
             separated,}
\end{equation*}
which precisely is finite propagation speed.

In this article we prove finite propagation speed when $\Delta$ is a self-adjoint
(s.a.) Laplace operator on a class of singular spaces, namely metric graphs. There
exists a whole family of such Laplace operators, for an extensive discussion see
\cite{KS1, KS8}. Previously and to the best of our knowledge finite propagation
speed on spaces with singularities has only been proved when the configuration space
has conical singularities \cite{Cheeger_Taylor}. As for other  applications we
mention that in the context of neuronal networks finite propagation speed on axons
has been discussed in \cite{Atay_Hutt}.

Recently one of the authors (R.S.) proved finite propagation speed for an arbitrary
s.a.\ Laplacian on star graphs (possibly having discrete eigenvalues) and on
arbitrary metric graphs under two restrictions : (i) $-\Delta\ge 0$, and (ii) at
least one of the points $p$ or $q$ is on one of the exterior edges \cite{Schrader}.
The proof used methods entirely different from the energy estimates usually employed
for the proof of finite propagation speed. It is based on properties of the (improper)
eigenfunctions of the Laplacians and their analytic properties as functions of the
spectral parameter. The proof we will give here, though only for a subclass of
Laplacians for which $-\Delta\ge 0$, is closer to the standard proof, which uses
energy estimates. The crucial new ingredient is an additional term in the standard
local energy functional and involves the boundary values at the vertices of the
graph for a given solution of the wave equation. Relevant for the proof of finite
propagation speed here as well as in \cite{Schrader} is that the self-adjoint
Laplacians are defined by \emph{local} boundary conditions, for details see
\cite{KS1,KS8}. In the usual contexts the self-adjointness of the Laplacian makes
the discussion of the existence and the uniqueness of solutions of the wave equation
for given $L^2$ Cauchy data relatively easy. The reason is that this self-adjointness
implies nice operator properties of the wave kernel
$W(t)$, which are easily obtained with help of the spectral
theorem. This is nicely worked out in \cite{Cheeger_Gromov_Taylor} and our presentation
has in a large part been motivated by the discussion given there. Then Sobolev inequalities
combined with the ellipticity of the Laplacian form
the tools for transforming $L^2$ properties of the solutions to analytic properties
like continuity and differentiability. Our discussion also uses (and needs) Sobolev
inequalities in order to control the boundary values since they enter the energy
functional. And the Laplacians we discuss have just this property that Sobolev
inequalities can be invoked. As a matter of fact, at the moment we do not know how
to deal with the other Laplacians as given and described in \cite{KS1,KS8}.

The article is organized as follows. In section \ref{sec:basics} we first recall
some basic facts about Laplacians on metric graphs and then we single out those we
shall mainly work with.
In section \ref{sec:uniqueness} we establish existence and
uniqueness of solutions of the wave equation for given Cauchy data. In section
\ref{sec:finprop} we introduce the local energy functional, which allows us to mimic
(and modify) standard proofs on finite propagation speed. The appendix provides the
Sobolev type estimates we need.

\section{Basic Structures}\label{sec:basics}

In this section we revisit the theory of Laplace operators on a metric graph $\cG$.
The material presented here is borrowed from the articles \cite{KS1}, \cite{KS8} and
\cite{KS9}.

A finite graph is a 4-tuple $\cG=(\cV,\cI,\cE,\partial)$, where $\cV$ is a finite set
of \emph{vertices}, $\cI$ is a finite set of \emph{internal edges}, $\cE$ is a
finite set of \emph{external edges}. For simplicity, from now on when we speak of a graph we will
mean a finite graph.

Elements in $\cI\cup\cE$ are called
\emph{edges}. The map $\partial$ assigns to each internal edge $i\in\cI$ an ordered
pair of (possibly equal) vertices $\partial(i):=(v_1,v_2)$ and to each external
edge $e\in\cE$ a single vertex $v$. The vertices $v_1=:\partial^-(i)$ and
$v_2=:\partial^+(i)$ are called the \emph{initial} and \emph{final} vertex of the
internal edge $i$, respectively. The vertex $v=\partial(e)$ is the initial vertex of
the external edge $e$. If $\partial(i)=(v,v)$, that is, $\partial^-(i)=\partial^+(i)$
then $i$ is called a \emph{tadpole}. To simplify the discussion, we will exclude tadpoles.
Two vertices $v$ and $v^\prime$ are called \emph{adjacent} if there is an internal
edge $i\in\cI$ such that $v\in\partial(i)$ and $v^\prime\in\partial(i)$. By definition
$\text{star}(v)\subseteq \cV$ of $v\in\cV$ is the set of vertices adjacent to $v$. A
vertex $v$ and the (internal or external) edge $j\in\cI\cup\cE$ are \emph{incident} if
$v\in\partial(j)$.

We do not require the map $\partial$ to be injective. In particular, any two
vertices are allowed to be adjacent to more than one internal edge and two different
external edges may be incident with the same vertex. If $\partial$ is injective and
$\partial^-(i)\neq\partial^+(i)$ for all $i\in\cI$, the graph $\cG$ is called
\emph{simple}. The \emph{degree} $\deg(v)$ of the vertex $v$ is defined as
\begin{equation*}
    \deg(v)=|\{e\in\cE\mid\partial(e)=v\}|+|\{i\in\cI\mid\partial^-(i)
        =v\}|+|\{i\in\cI\mid\partial^+(i)=v\}|,
\end{equation*}
that is, it is the number of (internal or external) edges incident with the given
vertex $v$ Throughout the whole work we will assume that the graph $\cG$ is
connected. In particular, this implies that any vertex of the graph $\cG$ has
nonzero degree, i.e., for any vertex there is at least one edge with which it is
incident.

The graph $\cG_{\mathrm{int}}=(\cV,\cI,\emptyset,\partial|_{\cI})$ will be called
the \emph{interior} of the graph $\cG=(\cV,\cI,$ $ \cE,\partial)$. It is obtained
from $\cG$ by eliminating all external edges $e$. Correspondingly, if $\cE\neq
\emptyset$, the graph
$\cG_{\mathrm{ext}}=(\partial\cV,\emptyset,\cE,\partial|_{\cE})$ is called the
\emph{exterior} of $\cG$. Here $\partial\cV\subseteq\cV$ is defined to be the set
consisting of those vertices $v$ which are of the form $v=\partial(e)$ for some
$e\in\cE$. We will view both $\cG_{\mathrm{int}}$ and $\cG_{\mathrm{ext}}$ as
subgraphs of $\cG$.

We will endow the graph with the following metric structure. Any internal edge
$i\in\cI$ will be associated with an interval $I_i=[0,a_i]$ with $a_i>0$ such that
the initial vertex of $i$ corresponds to $x=0$ and the final one to $x=a_i$. Any
external edge $e\in\cE$ will be associated with a half line $I_e=[0,+\infty)$. We
call the number $a_i$ the length of the internal edge $i$. We make the notational
convention that $a_e=\infty$ if $e\in\cE$. We will consider the set
$I_j,j\in\cE\cup\cI$ as a subset of $\cG$ and write $p\cong(j,x)$ for any point $p$
on $I_j$ with coordinate  $x$. The set of lengths $\{a_i\}_{i\in\cI}$, which will
also be treated as an element of $\R^{|\cI|}$, will be denoted by $\underline{a}$.
There is a canonical distance function $d(p,q),\;(p,q\in\cG)$ making the graph a metric
space. In particular $d(p,q)$ is continuous in both variables. So a graph $\cG$ endowed
with a metric structure $\underline{a}$ is called a \emph{metric graph}, denoted by
$(\cG,\underline{a})$. From now on the set $\underline{a}$ of lengths will be fixed
and we will simply speak of the metric space $\cG$. For given $p\in\cG$ and $t> 0$
let $B(p,t)$ denote the closed set of points in $\cG$ with distance from $p$ less or
equal to $t$. By definition its boundary $\partial B(p,t)$ is the set of points with
distance $t$ from $p$. Trivially $B(p,t)\subseteq B(p,t^\prime)$ for all
$t<t^\prime$ (with $B(p,t)\subset B(p,t^\prime)$ for all $t<t^\prime$ when $\cE\neq
\emptyset$) and
\begin{equation*}
    \lim_{t\uparrow \infty}B(p,t)=\bigcup_{0<t<\infty}B(p,t)=\cG.
\end{equation*}
The boundary set $\partial B(p,t)$ deserves special attention. As a function of $t$
the number of elements in $\partial B(p,t)$ is obviously piecewise constant. Here is
a partial list of properties. The number of elements in $\partial B(p,t)$ satisfies
\begin{align*}
    |\partial B(p,t)|
        =\begin{cases}
            2\qquad      &\mbox{for $p\in I_j\setminus \partial I_j$,
                                                $0<t<\dist(p,\partial I_j)$}\\
            \deg(p)\quad &\mbox{if $p$ is a vertex and $t<\dist(p,\text{star}(p))$}\\
            |\cE|\quad   &\max_{q\in\cG_{int}}d(p,q)<t.
         \end{cases}
\end{align*}
Boundaries at different times have vanishing intersection,
    \begin{equation*}\label{intersec}
\partial B(p,t)\cap\partial B(p,t^\prime)=\emptyset,\qquad t\neq t^\prime.
\end{equation*}
Figure \ref{fig:1} provides an example, which serves as a motivation for the
following definition.
\begin{definition}\label{def:coinc}
Given $p$ and $t$, a point $q\in\partial B(p,t)$ is a \emph{point of coincidence},
if for all $s<t$ sufficiently close to $t$ there are two
different points $q_l(s),q_r(s)\in \partial B(p,s)$ such that
\begin{equation*}
    \lim_{s\uparrow t}q_l(s)=\lim_{s\uparrow t}q_r(s)=q
\end{equation*}
holds. Let ${\bf Coin}(p,t)\subseteq \partial B(p,t)$ denote the subset of points of
coincidence. Given $p$, $t$ is \emph{critical} if the set ${\bf
Coin}(p,t)\cup(\partial B(p,t)\cap\cV)$ is non-empty. Given $p$, the set of
critical times $t>0$ is denoted by $\cT(p)$.
\end{definition}

Note that the set $\cT(p)$ contains the set of $t\ge 0$ at which $|\partial
B(p,t)|$ is discontinuous. $\cT(p)$ may be strictly larger.
As an example consider the case where $\cG$ is a star graph with two external edges and vertex $v$.
If $t$ is such that $v\in \partial B(p,t)$, that is $d(v,p)=t$,
then $|\partial B(p,t)|$ is continuous at $t$. More involved examples may easily
be constructed. ${\bf Coin}(p,t)\cap(\partial B(p,t)\cap\cV)$ may be non-empty. Also ${\bf
Coin}(p,t)\subset \cG_{int}$ and if ${\bf Coin}(p,t)\cap I_i\neq \emptyset$ for some
$t$ and $i\in\cI$, then ${\bf Coin}(p,t^\prime)\cap I_i= \emptyset$ for all
$t^\prime\neq t$. Similarly if $v\in\partial B(p,t)$, then $v\notin\partial
B(p,t^\prime)$ for all $t^\prime\neq t$. From these two observations one easily
deduces that $\cT(p)$ is a finite set with $|\cT(p)|\le |\cI|+|\cV|$.

Figure~\ref{fig:1} shows the example of a graph with two external edges $e_1$,
$e_2$, and two internal edges $i_1$, $i_2$ of equal length $a=a_{i_1}=a_{i_2}$.
There are two vertices $v_1$ and $v_2$. Consider a point $p$ on the edge $i_1$ with coordinate
$a/2$. The set $\partial B(p,t)$ consists of 2 points as long as $0<t\le a/2$, of
four points when $a/2<t<a$, of three points when $t= a$, and of two points, when
$t>a$.
Thus $\partial B(p,t)=\{q_l(t),q_r(t),q_{e_1}(t),q_{e_2}(t)\}$ when $a/2<t<a$ and
$\partial B(p,t)=\{q_{e_1}(t),q_{e_2}(t)\}$ when $t>a$. The two points $q_l(t)$ and
$q_r(t)$ at a distance $a/2<t<a$ from $p$, lie on the edge $i_2$, and collapse to an
\emph{antipodal point} $q$ (with coordinate $a/2$) of $p$, when $t$ increases to
$a$. So ${\bf Coin}(p,a)=\{q\}$ holds, while ${\bf Coin}(p,t)=\emptyset$ for all
$t\neq a$.
\begin{figure}[htb]
\begin{center}
\includegraphics{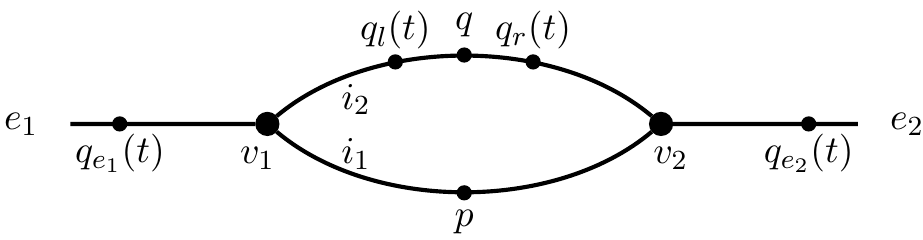}
\caption{A point $q$ of coincidence in $\partial B(p,t)$}\label{fig:1}
\end{center}
\end{figure}

In Riemannian geometry there is an analogue to the notion of a point of coincidence.
It arises in the context of geodesics and is given by the notion of a conjugate
point. Thus a time $t$, for which ${\bf Coin}(p,t)\neq \emptyset$ while
${\bf Coin}(p,t^\prime)= \emptyset$ for all $t^\prime<t$,
is the analogue of the injectivity radius, that is the radius at which
the exponential map ceases to be injective.

There is a canonical Lebesgue measure $\cG$, so that the notion of $L^p(\cG)$ spaces
of measurable functions on $\cG$ makes sense. More generally, we will consider the
spaces $L^p(\cF)$ where $\cF$ is any measurable subset of $\cG$ and use the notation
$$
    \int_{\cF}\psi(p)\,dp
$$
to describe the integral of an element $\psi\in L^1(\cF)$ and the notation
$$
    \langle \varphi,\psi\rangle_\cF=\int_{\cF}\overline{\varphi(p)}\,\psi(p)\,dp.
$$
to describe the scalar product of two elements $\varphi,\psi$ in the Hilbert space
$L^2(\cF)$. Also we write $\|\psi\|^2_\cF=\langle \psi,\psi\rangle_\cF$. Whenever
the context is clear we will simply write $\|\psi\|^2$ and $\langle
\varphi,\psi\rangle$ for $\|\psi\|^2_\cG$ and $\langle \varphi,\psi\rangle_\cG$
respectively. There is an alternative way to obtain $L^2(\cG)$, which is useful for
the discussion of Laplace operators. The central idea is to consider for any measurable
function $\psi$ on $\cG$ its restriction $\psi_i$ to the edge
$I_i,\;i\in\cE\cup\cI$.

So consider the Hilbert space
\begin{equation*}\label{hilbert}
    \cH\equiv\cH(\cE,\cI,\underline{a})=\cH_{\cE}\oplus\cH_{\cI},\qquad
    \cH_{\cE}=\bigoplus_{e\in\cE}\cH_{e},\qquad
    \cH_{\cI}=\bigoplus_{i\in\cI}\cH_{i},
\end{equation*}
where $\cH_{e}=L^2([0,\infty))$ for all $e\in\cE$ and $\cH_i=L^2([0,a_i])$ for all
$i\in\cI$. Then $L^2(\cG)\cong \cH$ holds and from now on we shall interchangeably
work with both notations. Moreover, to keep our notation simple, we shall identify
$I_i$ with the interval $[0,+\infty)$ if $i\in\cE$ and with $[0,a_i]$ if $i\in\cI$,
unless there is danger of confusion. Of course the spaces $L^p(\cG)$ have a
similar alternative description.

By $\cD_i$ with $i\in\cE\cup\cI$ we denote the set of all $\psi_i\in\cH_i$ such that
$\psi_i$ and its derivative $\psi^\prime_i$ are absolutely continuous, and its
second derivative $\psi^{\prime\prime}_i$ is square integrable. Let $\cD_i^0$
denote the set of those elements $\psi_i\in\cD_i$ which satisfy
\begin{equation*}
\begin{matrix}
\psi_i(0)=0\\[1ex] \psi^\prime_i(0)=0
\end{matrix} \quad \text{for $i\in\cE$, and}\quad
\begin{matrix}
\psi_i(0)=\psi_i(a_i)=0\\[1ex]
\psi^\prime_i(0)=\psi^\prime_i(a_i)=0
\end{matrix}
\quad\text{for $i\in\cI$}.
\end{equation*}
Let $\Delta^0$ be the differential operator
\begin{equation}\label{Delta:0}
\left(\Delta^0\psi\right)_i (x) = \psi^{\prime\prime}_i(x),\qquad
	x\in I_i,\,i\in\cI\cup\cE,
\end{equation}
with domain
\begin{equation*}
\cD^0=\bigoplus_{i\in\cE\cup\cI} \cD_i^0 \subset\cH.
\end{equation*}
It is straightforward to verify that $\Delta^0$ is a closed symmetric operator with
deficiency indices equal to $|\cE|+2|\cI|$.

We introduce an auxiliary finite-dimensional Hilbert space
\begin{equation}\label{K:def}
\cK\equiv\cK(\cE,\cI)=\cK_{\cE}\oplus\cK_{\cI}^{(-)}\oplus\cK_{\cI}^{(+)}
\end{equation}
with $\cK_{\cE}\cong\C^{|\cE|}$ and $\cK_{\cI}^{(\pm)}\cong\C^{|\cI|}$. Let
${}^d\cK$ denote the ``double'' of $\cK$, that is, ${}^d\cK=\cK\oplus\cK$.

For any
\begin{equation*}\label{domain}
\displaystyle\psi\in\cD:=\bigoplus_{i\in\cE\cup\cI} \cD_i
\end{equation*}
we set
\begin{equation*}\label{lin1}
[\psi]:=\underline{\psi}\oplus \underline{\psi}^\prime\in{}^d\cK,
\end{equation*}
with the boundary values $\underline{\psi}$ and $\underline{\psi}^\prime$ defined by
\begin{equation*}\label{lin1:add}
\begin{split}
	\underline{\psi}
&= \bigl( (\psi_e,\,e\in\cE),(\psi_i(0),\,i\in\cI),(\psi_i(a_i),\,i\in\cI)\bigr)^t,\\
	\underline{\psi}'
&= \bigl( (\psi_e',\,e\in\cE),(\psi_i'(0),\,i\in\cI),(-\psi_i'(a_i),\,i\in\cI)\bigr)^t.\\
\end{split}
\end{equation*}
Here the superscript $t$ denotes transposition. Let $J$ be the canonical symplectic
matrix on ${}^d\cK$,
\begin{equation*}\label{J:canon}
J=\begin{pmatrix} 0& \1 \\ -\1 & 0
\end{pmatrix}
\end{equation*}
with $\1$ being the identity operator on $\cK$. Consider the non-degenerate
Hermitian symplectic form
\begin{equation*}\label{omega:canon}
\omega([\varphi],[\psi]) := \langle[\varphi], J[\psi]\rangle,
\end{equation*}
where $\langle\cdot,\cdot\rangle$ denotes the inner product in ${}^d
\cK\cong\C^{2(|\cE|+2|\cI|)}$.

A linear subspace $\cM$ of ${}^d\cK$ is called \emph{isotropic} if the form $\omega$
vanishes on $\cM$ identically. An isotropic subspace is called \emph{maximal} if it
is not a proper subspace of a larger isotropic subspace. Every maximal isotropic
subspace has complex dimension equal to $|\cE|+2|\cI|$.

Let $A$ and $B$ be linear maps of $\cK$ onto itself. By $(A,B)$ we denote the linear
map from ${}^d\cK=\cK\oplus\cK$ to $\cK$ defined by the relation
\begin{equation*}
(A,B)\; (\chi_1\oplus \chi_2) := A\, \chi_1 + B\, \chi_2,
\end{equation*}
where $\chi_1,\chi_2\in\cK$. Set
\begin{equation*}\label{M:def}
\cM(A,B) := \Ker\, (A,B).
\end{equation*}
\begin{theorem}\label{theo:sa}
A subspace $\cM\subset{}^d\cK$ is maximal isotropic if and only if there exist
linear maps $A,\,B:\; \cK\rightarrow\cK$ such that $\cM=\cM(A,B)$ and
\begin{equation}\label{abcond}
\begin{split}
\mathrm{(i)}\; & \;\text{the map $(A,B)\;:\;{}^d\cK\rightarrow\cK$ has maximal
rank equal to
$|\cE|+2|\cI|$,}\qquad\\
\mathrm{(ii)}\; &\;\text{$AB^{\dagger}$ is self-adjoint,
    $AB^{\dagger}=BA^{\dagger}$.}
\end{split}
\end{equation}
\end{theorem}
A proof is given in \cite{KS8}. The boundary conditions $(A,B)$ and
$(A^\prime,B^\prime)$ satisfying \eqref{abcond} are called equivalent if the
corresponding maximal isotropic subspaces coincide, that is,
$\cM(A,B)=\cM(A^\prime,B^\prime)$, and this in turn holds if and only if there is an
invertible $C$ such that $A^\prime=CA,B^\prime=CB$ is valid.

There is a one-to-one correspondence between all self-adjoint extensions of
$\Delta^0$ and maximal isotropic subspaces $\cM$ of ${}^d\cK$ (see \cite{KS1},
\cite{KS8}). In explicit terms, any self-adjoint extension of $\Delta^0$ is the
differential operator defined by \eqref{Delta:0} with domain
\begin{equation}\label{thru}
    {\rm Dom}(\Delta)=\{\;\psi\in\cD\;|\; [\psi]\in\cM\;\},
\end{equation}
where $\cM$ is a maximal isotropic subspace of ${}^d\cK$. Conversely, any maximal
isotropic subspace $\cM$ of ${}^d\cK$ defines through \eqref{thru} a self-adjoint
operator $\Delta_\cM $. In the sequel we will call the operator $\Delta_\cM$ a
Laplace operator on the metric graph $\cG$. Thus we have $\Delta_\cM
\psi=\psi^{\prime\prime}$ and in particular
\begin{equation}\label{normprimeprime}
    \|\psi^{\prime\prime}\|
        =\|\Delta_\cM \psi\| \qquad \mbox{for $\psi\in\cD(\Delta_\cM)$}.
\end{equation}

From the discussion above it follows immediately that any self-adjoint Laplace
operator on $\cH$ equals $\Delta_\cM$ for some maximal isotropic subspace $\cM$.
Moreover, $\Delta_\cM=\Delta_{\cM^\prime}$ if and only if $\cM=\cM^\prime$. For
short we will henceforth call $\cM$ a boundary condition. The role of the hermitian
symplectic form $\omega$ is clarified by the following observation. Consider the
hermitian symplectic form $\widehat{\omega}$ on $\cD$
$$
    \widehat{\omega}(\varphi,\psi)=(\Delta^0\varphi,\psi)-(\varphi,\Delta^0\psi).
$$
Then by Green's theorem
$$
    \widehat{\omega}(\varphi,\psi)=\omega([\varphi],[\psi])
$$
holds, such that $\widehat{\omega}$ vanishes on ${\rm Dom}(\Delta_\cM)$.

All operators $-\Delta_\cM$ are finite rank perturbations of each other and in
particular bounded from below. So any $-\Delta_\cM$ has absolutely continuous spectrum, equal to the
positive real axis and with multiplicity $|\cE|$. By definition the boundary condition $\cM$ is
\emph{real} if there are real matrices $A$ and $B$ such that $\cM=\cM(A,B)$. For
real $\cM$, the Laplacian $-\Delta_\cM$ is also real in the sense that for all
$\psi\in {\rm Dom}(-\Delta_\cM)$ also $\overline{\psi}\in {\rm Dom}(-\Delta_\cM)$
and $\overline{-\Delta_\cM\psi}=-\Delta_\cM\overline{\psi}$. For more details, see
\cite{KS1,KS8}.

For given $\cM=\cM(A,B)$ the orthogonal projector $P_\cM$ in $^d\cK$ onto $\cM$ is given as
\begin{equation*}\label{projection}
\begin{split}
P_{\cM} & = \begin{pmatrix} -B^\dagger \\ A^\dagger
\end{pmatrix}(A A^\dagger + B B^\dagger)^{-1} (-B, A)\\[1ex]
& = \begin{pmatrix} B^\dagger (A A^\dagger + B
B^\dagger)^{-1} B & - B^\dagger (A A^\dagger + B B^\dagger)^{-1} A \\
-A^\dagger (A A^\dagger + B B^\dagger)^{-1} B & A^\dagger (A A^\dagger + B
B^\dagger)^{-1} A
\end{pmatrix},
\end{split}
\end{equation*}
where the block matrix notation is used with respect to the orthogonal decomposition
${}^d\cK = \cK\oplus\cK$. With the same decomposition define
$$
\Omega=\begin{pmatrix} 0 & \1 \\ 0 & 0
\end{pmatrix}
$$
and set $\Omega_\cM=P_{\cM} \Omega P_{\cM} $, giving
\begin{equation}\label{OmegaM}
\Omega_{\cM} = - \begin{pmatrix} -B^\dagger \\ A^\dagger \end{pmatrix} (A
A^\dagger + B B^\dagger)^{-1} A B^\dagger (A A^\dagger + B B^\dagger)^{-1}
(-B,A),
\end{equation}
 a hermitian $2(|\cE|+2|\cI|)\times2(|\cE|+2|\cI|)$ matrix. Observe that $\Omega$ is
\emph{half} of the canonical symplectic matrix $J$ in the sense that
$J=\Omega-\Omega^\dagger$ holds. Now
$\Omega_\cM=\Omega_{\cM}^\dagger=P_{\cM}\Omega^\dagger P_{\cM}$ and hence $P_{\cM}J
P_\cM=0$, another way of stating that the space $\cM$ is isotropic.

We quote the following result from \cite{KS9}.
\begin{proposition}\label{prop:sesq}
For any maximal isotropic subspace $\cM\subset{}^d\cK$ the identity
\begin{equation}\label{sesquilinear:2}
\langle\varphi, -\Delta_\cM \psi\rangle_{\cG} =
\langle\varphi^\prime,\psi^\prime\rangle_{\cG}
+\langle[\varphi], \Omega_{\cM}[\psi]\rangle_{{}^d\cK}
\end{equation}
holds for all $\varphi,\psi\in {\rm Dom}(-\Delta_\cM )$.
\end{proposition}
Observe that by the identity \eqref{sesquilinear:2} $||\psi^\prime||_\cG$ is finite
for any $\psi\in {\rm Dom}(-\Delta_\cM )$. Indeed, the boundary values $[\psi]$ and
$[\psi]$ exist so that $\langle[\psi], \Omega_{\cM}[\psi]\rangle_{{}^d\cK}$ is well
defined and finite because ${\rm Dom}(-\Delta_\cM )\subset\cD$. This proposition
immediately gives the first part of
\begin{corollary}\label{cor:qpos}
If the boundary condition $\cM$ is such that $\Omega_\cM\ge 0$, then also $-\Delta_\cM\ge 0$.
If $\Omega_\cM> 0$, then $0$ is not an eigenvalue of $-\Delta_\cM $.
\begin{proof}
To prove the second part, assume there is $\psi\in {\rm Dom}(-\Delta_\cM)$ with
$-\Delta_\cM \psi=0$. \eqref{sesquilinear:2} gives $\psi^\prime=0$. So $\psi$ has to be
constant on each edge. But \eqref{sesquilinear:2} also implies $[\psi]=0$ which is only
possible, if $\psi=0$.
\end{proof}
\end{corollary}
The converse does not hold, that is $-\Delta_\cM\ge 0$ does not imply $\Omega_\cM\ge 0$,
as Example~3.8 in \cite{KPS2} shows.

The next corollary is a trivial consequence of \eqref{OmegaM} and \eqref{sesquilinear:2}
in combination with Theorem \ref{theo:sa}.
\begin{corollary}\label{cor:psiprime}
If the boundary condition $\cM$ is such that $\Omega_\cM\ge 0$ and hence $-\Delta_\cM\ge 0$
is valid, then
\begin{equation}\label{psiprimebound1}
||\psi^\prime||_{\cG}\le ||\sqrt{-\Delta_\cM}\psi||_\cG
\end{equation}
holds.

If $\Omega_\cM>0$ and if the boundary value $[\psi]$ is non-vanishing, then the
inequality is strict. $\Omega_\cM= 0$ if and only if $\cM=\cM(A,B)$ is such that
$AB^\dagger=0$ holds and then \eqref{psiprimebound1} is actually an equality for all
$\psi\in {\rm Dom}(\sqrt{-\Delta_\cM})$.
\end{corollary}
The inequality \eqref{psiprimebound1} for the norm of the first derivative compares
with the identity \eqref{normprimeprime} for the second derivative.
Characterizations of maximal isotropic subspaces $\cM(A,B)$ satisfying
$AB^\dagger=0$ are given in \cite{KPS1}, Proposition 2.4 and \cite{KS8} Remark 3.9.
There also examples are provided.

With respect to the decomposition \eqref{K:def} any vector $\chi$ in $\cK$ can be
represented as
\begin{equation}\label{chi}
	\chi = \bigl((\chi_e,\,e\in\cE), (\chi^-_i,\,i\in\cI), (\chi^+_i,\,i\in\cI)\bigr)^t.
\end{equation}
Consider the orthogonal decomposition
$$
\cK=\bigoplus_{v\in\cV}\cL_v
$$
with $\cL_v$ being the linear subspace of dimension $\deg(v)$ spanned by those
elements $\chi$ in $\cK$ of the form \eqref{chi} which satisfy
\begin{equation*}\label{chiv}
\begin{split}
    \chi_e   &=0,\quad \mbox{if $e\in\cE$ is not incident with $v$,}\\
    \chi^-_i &=0,\quad \mbox{if $v$ is not an initial vertex of $i\in\cI$,}\\
    \chi^+_i &=0,\quad \mbox{if $v$ is not a final vertex of $i\in\cI$.}
\end{split}
\end{equation*}
Set $^d \cL_v:=\cL_v\oplus\cL_v\cong \C^{2\deg(v)}$. Obviously each $^d \cL_v$
inherits in a canonical way a symplectic structure from $^d\cK$ such that the
orthogonal and symplectic decomposition
$$
    \bigoplus_{v\in\cV}{}^d\cL_v={}^d\cK
$$
holds. If the boundary condition $\cM=\cM(A,B)$ is \emph{local} (see \cite{KS8} for
more details), then $A$ and $B$ have a decomposition
$$
    A=\bigoplus_{v\in\cV}A_v,\qquad B=\bigoplus_{v\in\cV}B_v,
$$
that is $A_v$ and $B_v$ are linear transformations on $\cL_v$, such that $(A_v,B_v)$
is a linear transformation in $^d\cL_v$. So correspondingly there is a decomposition
$$
    \cM=\bigoplus_{v\in\cV}\cM_v
$$
with $\cM_v=\cM(A_v,B_v)=\ker(A_v,B_v)$. Let $Q_v$ denote the orthogonal projection
in $\cK$ onto $\cL_v$ and $^dQ_v:=Q_v\oplus Q_v$ its double, that is the orthogonal
projection in $^d\cK$ onto $^d\cL_v$. Then we have
\begin{lemma}\label{lem:comm} The relation
\begin{equation}\label{pv}
    ^dQ_v\,P_\cM=P_\cM\;^dQ_v
\end{equation}
holds and equals the orthogonal projection $P_v$ in $^d\cK$ onto $\cM_v$.
\end{lemma}

Note that \eqref{pv} implies that $^dQ_v\,P_\cM$ is an orthogonal projector. Also
$P_vP_{v^\prime}=P_{v^\prime}P_v$, so in particular the $P_v$'s commute pairwise.

\begin{proof}
Let $\cP_{\cM_v}$ denote the orthogonal projection in $^d\cL_v$ onto $\cM_v$, that
is $\cP_{\cM_v}$ is obtained in a similar way from $(A_v,B_v)$ as is $\cP_{\cM}$
from the pair $(A,B)$ with $^d\cK$ being replaced by $^d\cL_v$. Denote by $P_v$ the
orthogonal projection in $^d\cK$ onto $\cM_v$. Then $P_{\cM_v}$ is the restriction
of $\cP_{\cM}$ to $^d\cL_v$, $P_{\cM_v}=\cP_{\cM}|_{^d\cL_v}$. Similarly we view $^d
Q_v$ as a map from $^d\cK$ onto $^d\cL_v$ and its restriction ${}^dQ_{v}|_\cM$ to
$\cM$ as map from $\cM$ onto $\cM_v$. Then we have a commutative diagram
$$
\begin{array}{ccc}
\quad\,^d\cK&\mathop{\put(-10,4){\vector(2,0){50}}}\limits^{\qquad\cP_\cM}& \cM\\
&&\\
\put(-20,0){$^{{}^dQ_v}$}\put(5,20){\vector(0,-1){30}}&&
\put(35,20){\vector(0,-1){30}}\qquad\qquad ^{ {}^dQ_{v}|_{\cM}}\\
&&\\
\quad\, ^d\cL_v&\mathop{\put(-10,4){\vector(2,0){50}}}\limits^{\qquad\cP_{\cM_v}
=\cP_{\cM}|_{^d\cL_v}}&\cM_v
\end{array}
$$
from which \eqref{pv} and the equality $P_v={}^dQ_v\,P_\cM$ follow.
\end{proof}

Correspondingly we obtain
\begin{equation*}\label{summ}
    \Omega_\cM=\bigoplus_{v\in\cV} \Omega_{v}
\end{equation*}
with $\Omega_{v}=P_v\Omega=\Omega P_v=P_v\Omega P_v$. Since $\cM_v$ is a subspace of
$\cM$ also
\begin{equation}\label{pmv}
    P_v\cP_{\cM}=\cP_{\cM}P_v=^dQ_v\cP_{\cM}=\cP_{\cM}\;^dQ_v=P_v
\end{equation}
holds. For any subset $\cV^\prime$ of $\cV$ we introduce the orthogonal projection
$P_{\cV^\prime}=\bigoplus_{v\in\cV^\prime} P_v$. These $P_{\cV^\prime}$'s commute
pairwise. Finally set
\begin{equation}\label{qv}
\Omega_{\cM,\,\cV^\prime}=P_{\cV^\prime}\Omega_{\cM}
=P_{\cV^\prime}\Omega_{\cM}P_{\cV^\prime}
=\Omega_{\cM}P_{\cV^\prime}
\end{equation}
such that $\Omega_{\cM,\,\cV}=\Omega_\cM$ holds. More generally
\begin{equation*}\label{qp}
    \Omega_{\cM,\,\cV^{\prime\prime}}
        =P_{\cV^{\prime\prime}}\Omega_{\cM,\,\cV^\prime}P_{\cV^{\prime\prime}}
\end{equation*}
is valid for any pair $\cV^{\prime\prime}\subseteq \cV^\prime$. The following lemma
is trivial
\begin{lemma}\label{lem:qs}
$\Omega_\cM\ge 0$ is valid if and only if $\Omega_{\cM,\,v}\ge 0$ for all $v\in\cV$. Similarly
$\Omega_\cM>0$ holds if and only if $\Omega_{\cM,\,v}> 0$ for all $v\in\cV$.
If $\Omega_\cM\ge 0$ then $0\le \Omega_{\cM,\,\cV^{\prime\prime}}\le \Omega_{\cM,\,\cV^{\prime}}$ for all
$\cV^{\prime\prime}\subseteq \cV^\prime$.
\end{lemma}

\section{Existence and Uniqueness of Solutions of the Wave Equation}\label{sec:uniqueness}

Throughout this section we fix a maximal isotropic subspace $\cM$ of $^d\cK$.
We introduce the wave kernel
\begin{equation*}\label{wavekernel}
    W(t)=W_\cM(t)=\frac{\sin\bigl(\sqrt{-\Delta_\cM}t\bigr)}{\sqrt{-\Delta_\cM}},\qquad t\in\R
\end{equation*}
which is defined through the spectral representation of the self-adjoint operator
$-\Delta_\cM$ and operator calculus. Note that at the moment we do not (need to)
assume $-\Delta_\cM$ to be non-negative. $W(t)$ is bounded and self-adjoint for all
$t\in \R$ with $W(0)=0$. Set
\begin{equation*}\label{infspec}
    \rho_\cM(t)=\cosh\bigl(t \sqrt{-\varepsilon_\cM}\bigr)
\end{equation*}
with
\begin{equation*}\label{infspec1}
    \varepsilon_\cM=\min\bigl(\inf\spec(-\Delta_\cM),0\bigr).
\end{equation*}
$\rho_\cM(t)=1$ for all $t\in\R$, if $-\Delta_\cM$ is non-negative.
$W(t)$ is a bounded operator and norm continuous
in $t$
\begin{equation}\label{wtnorm}
    \|W(t)\|\le |t|\,\rho_\cM(t),\qquad \|W(t)-W(t^\prime)\|
\le |t-t^\prime|\max(\rho_\cM(t),\rho_\cM(t^\prime)).
\end{equation}
We will also consider the time derivatives of the wave kernel $W(t)$:
\begin{equation}\label{wavekernel1}
\begin{split}
    \partial_t W(t)   &= \cos\bigl(\sqrt{-\Delta_\cM}t\bigr),\\[.5ex]
    \partial_t^2 W(t) &= -\sqrt{-\Delta_\cM}\sin\bigl(\sqrt{-\Delta_\cM}t\bigr)
							= \Delta_\cM W(t),
\end{split}
\end{equation}
where the derivatives are taken in the strong operator topology.
$\partial_t W(t)$ is a bounded self-adjoint operator for all $t$ with
operator norm bound
\begin{equation}\label{twtnorm}
    \|\partial_t W(t)\|\le \rho_\cM(t),\qquad t\in\R.
\end{equation}
The estimates \eqref{wtnorm} and \eqref{twtnorm} follow from the spectral theorem and
the trivial bounds
$$
\sup_{x\in\R}\left|\frac{\sin x}{x}\right|\le 1,
                \;\sup_{0\le x\le y}\frac{\sinh x}{x}\le \cosh y,\; y\ge 0.
$$
It will be convenient to introduce the following notation. Choose $m^2\ge 0$ such
that $-\Delta_{\cM}+m^2$ is non-negative. In particular, if $-\Delta_{\cM}\ge 0$,
we set $m^2=0$.

For any $\ga\ge 0$ the domain $\text{Dom}\bigl((-\Delta_{\cM}+m^2)^{\ga/2}\bigr)$ may be
equipped with the inner product
$$
\langle\langle\phi,\psi\rangle\rangle=\langle\phi,\psi\rangle_\cG
+\langle(-\Delta_{\cM}+m^2)^{\ga/2}\phi,(-\Delta_{\cM}+m^2)^{\ga/2}\psi\rangle_\cG,
$$
turning it into a Hilbert, which we denote by $H^\ga_{\cM,m^2}\subset L^2(\cG)$. The relation
$H^\ga_{\cM,m^2}\subset H^{\ga^\prime}_{\cM,m^2}$ whenever $\ga^\prime\le \ga$ is obvious.
By construction the semi-norm on $H^\ga_{\cM,m^2}$
$$
\|\psi \|_{\cM,\ga} =\|(-\Delta_\cM+m^2)^{\ga/2}\psi \|_{\cG}
$$
satisfies $\|\psi \|_{\cM,\ga}^2\le \langle\langle\psi,\psi\rangle\rangle$ and is hence
a continuous map from $H^\ga_{\cM,m^2}$ onto the non-negative numbers.
It is a norm if and only if 0 is not am eigenvalue of $-\Delta_\cM+m^2$.
For varying $\ga$ these Sobolev (semi-)norms will constitute the basic tools when
we estimate solutions of the wave equation in terms of the initial data and to which we turn now.

For $k\in\N$, an iteration of~\eqref{wavekernel1} yields
\begin{equation}\label{wavekernelg}
\begin{split}
	\partial_t^{2k+1} W(t) &= \Delta_\cM^k \cos\bigl(\sqrt{-\Delta_\cM}t\bigr)
							= \Delta_\cM^k \partial_t W(t),\\
	\partial_t^{2k} W(t) &= (-1)^k \bigl(\sqrt{-\Delta_\cM}\bigr)^{2k-1}
							\sin\bigl(\sqrt{-\Delta_\cM}t\bigr)
= \Delta_\cM^k W(t).
\end{split}
\end{equation}
For $n\in\Z$ set $n_+ = \max(n,0)$. An easy application of the spectral
theorem gives the following
\begin{lemma}	\label{lem_wkernel}
{\ }
\begin{enumerate}\aitem
	\item{For all $n\in\N_0$, $t\in\R$, $\partial_t^n W(t)$ is a self-adjoint
		  operator with domain 	
		  \begin{equation*}\label{Wkdomain}
    			\text{Dom}\bigl(\partial_t^n W(t)\bigr)
                          = H^{(n-1)_+}_{\cM,m^2}\;,
		  \end{equation*}
		  commuting with $\Delta_\cM$, and mapping $H^{(l+n-1)_+}_{\cM,m^2}$ into
		  $H^{l}_{\cM,m^2}$, $l\in\N_0$. Moreover, $\partial_t^n W(t)$ is strongly
		  continuous in $t$ as an operator on $H^{(n-1)_+}_{\cM,m^2}$.}
	\item{For all $n\in\N_0$, $t\in\R$, the relation
		  \begin{equation*}	\label{WE}
				\Box_\cM \partial_t^n W(t) = 0
		  \end{equation*}	
		  holds on $H^{n+1}_{\cM,m^2}$, where $\Box_\cM$ is the \emph{d'Alembert
		  operator} associated with $\Delta_\cM$:	
		  \begin{equation*}	\label{DA}
				\Box_\cM = \partial_t^2 - \Delta_\cM.
		  \end{equation*}}			
\end{enumerate}
\end{lemma}
For any \emph{Cauchy data} $(\psi_0,\dot{\psi}_0)$ in $L^2(\cG)\times L^2(\cG)$ set
\begin{equation}\label{wave1}
    \psi(t)=\partial_t W(t)\,\psi_0 + W(t)\,\dot{\psi}_0,\qquad t\in\R,
\end{equation}
which is a well-defined element in $L^2(\cG)$ for all $t$, strongly
continuous in $t$. For what follows it will be
convenient to introduce the notation
\begin{equation*}
	H^{\alpha,\beta}_{\cM,m^2} = H^\alpha_{\cM,m^2}\oplus H^\beta_{\cM,m^2},\qquad
								\alpha,\,\beta\ge 0.
\end{equation*}
If the Cauchy data $(\psi_0,\dot\psi_0)$ belong to $H^{2,1}_{\cM,m^2}$,
then we obtain from Lemma~\ref{lem_wkernel} that for all $t$, $\psi(t)$
is twice strongly differentiable in $t$, that it belongs to $H^2_{\cM,m^2}$,
and that it is a solution of the initial value problem of the wave equation
\begin{equation}\label{wave3}
\begin{split}
    \Box_\cM\psi(t)    &= 0,\\
          \psi(t=0)    &= \psi_0,\\
    \partial_t\psi(t=0)&= \dot{\psi}_0.
\end{split}
\end{equation}
$\psi(t)\in H^2_{\cM,m^2}$ implies that $\psi(t)$ and $\psi(t)'$ are
absolutely continuous on every open edge for all
$t$. Also the boundary values
$[\psi(t)]$ at the vertices of $\cG$ may be taken.
If $(\psi_0,\dot\psi_0)\in H^{3,2}_{\cM,m^2}$, then an analogous
statement is true for $\partial_t \psi(t)$, and both $\psi(t)$ and $\partial_t \psi(t)$
have for all $t$ second order spatial derivatives on the open edges, which define elements
in $L^2(\cG)$. Therefore in this case the set of boundary values $[\partial_t\psi(t)]$
exists, too. If both the boundary condition $\cM$ and the Cauchy data are chosen
to be real, then $\psi(t)$ is real for  all times $t$. We also remark
that~\eqref{wave1} extends to
\begin{equation*}\label{wave4}
    \psi(t)=W(t-s)\,\partial_s \psi(s) - \bigl(\partial_s W(t-s)\bigr)\,\psi(s),
\end{equation*}
valid for all $t$, $s\in\R$.

Similar arguments lead to the following slightly more general result which
will be useful below:
\begin{proposition}\label{prop:waveeq}
Suppose that $\psi$ is defined by equation~\eqref{wave1} with Cauchy data
$(\psi_0,\dot\psi_0)$.
\begin{enumerate}	\aitem
\item If $(\psi_0,\dot\psi_0)\in H^{n+l, n+l-1}_{\cM,m^2}$, $n\in\N_0$, $l\in\N$,
	  then $\partial_t^n\psi(t)$, is of the form~\eqref{wave1} with
	  $\partial_t^n\psi(t)\in H^l_{\cM,m^2}$ for all $t$, and
	  $\partial_t^n \psi(t)$ is $l$ times strongly continuously differentiable in $t$.
	  If $l\ge 2$, then $\partial_t^n\psi(t)$ is a solution of the wave equation
	  with Cauchy data $(\Delta^k \psi_0,\Delta^k\dot\psi_0)$ if $n=2k$, and with
	  $(\Delta^k \dot\psi_0,\Delta^{k+1}\psi_0)$, respectively, if $n=2k+1$.
	  Furthermore, if $l\ge 3$, then the set $[\partial_t^n\psi(t)]$ of boundary
	  values of $\partial_t^n\psi(t)$ at the vertices of $\cG$ exists.
\item If $(\psi_0,\dot\psi_0)\in H^{n+2, n+1}_{\cM,m^2}$, $n\in\N_0$, then
	  $(-\Delta_\cM + m^2)^{n/2}\psi(t)$ is a solution of the wave equation
	  of the form~\eqref{wave1} with Cauchy data $((-\Delta_\cM+m^2)^{n/2}\psi_0,
	  (-\Delta+m^2)^{n/2}\dot\psi_0)$. If  $(\psi_0,\dot\psi_0)\in H^{n+3, n+2}_{\cM,m^2}$,
	  $n\in\N_0$, then the boundary values $[(-\Delta+m^2)^{n/2}\psi(t)]$ are
	  well-defined for all $t$.
\end{enumerate}
\end{proposition}

With~\eqref{wtnorm} and~\eqref{twtnorm}, the corresponding estimates in terms of the
Cauchy data are given by

\begin{proposition}\label{prop:apriori}
The following \emph{a priori} estimates are valid for $\psi(t)$, as defined by
\eqref{wave1} with Cauchy data $(\psi_0,\dot{\psi}_0)$:
\begin{enumerate}	\aitem
\begin{subequations}
	\item Suppose that $(\psi_0,\dot\psi_0)\in H^{n+2k,n+2k}_{\cM,m^2}(\cG)$, $k$, $n\in\N_0$.
			Then for all $t$
			\begin{equation*}	\label{estimatesa}
				\|\partial_t^{2k}\psi(t)\|_{\cM,n}
					\le \rho_{\cM}(t)\,\bigl(\|\Delta_\cM^k\psi_0\|_{\cM,n}
						+ |t|\,\|\Delta_\cM^k\dot\psi_0\|_{\cM,n}\bigr)
			\end{equation*}
			holds.
	\item Suppose that  $(\psi_0,\dot\psi_0)\in H^{n+2k+2,n+2k}_{\cM,m^2}(\cG)$,
			$k$, $n\in\N_0$. Then for all $t$
			\begin{equation*}	\label{estimatesb}
				\|\partial_t^{2k+1} \psi(t)\|_{\cM,n}
					\le \rho_{\cM}(t)\,\bigl(\|\Delta_\cM^k\dot\psi_0\|_{\cM,n}
						+ |t|\,\|\Delta_\cM^{k+1}\psi_0\|_{\cM,n}\bigr)
			\end{equation*}
			holds.
\end{subequations}
\end{enumerate}
\end{proposition}

Assume that $(\psi_0,\dot\psi_0)\in H^{2,2}_{\cM,m^2}$, so that by Proposition~\ref{prop:waveeq}
$\partial_t\psi(t)$ is strongly continuously differentiable in $t$. Hence we can write
\begin{equation*}
	(\partial_t\psi)(t_1) - (\partial_t\psi)(t_2)
		= \int_{t_2}^{t_1} \bigl(\partial_t^2\psi\bigr)(s)\,ds,
\end{equation*}
as a relation in the Hilbert space $L^2(\cG)$, where the last integral is a
Bochner integral. With the estimates given in Proposition~\ref{prop:apriori} we therefore
find
\begin{equation*}
\begin{split}
	\|(\partial_t&\psi)(t_1) - (\partial_t\psi)(t_2)\|\\
		 &\le |t_1-t_2|\,\max\bigl(\rho_\cM(t_1),\rho_\cM(t_2)\bigr)\,
				\bigl(\|\Delta_\cM \psi_0\| + \max(|t_1|, |t_2)\,\|\Delta_\cM\dot\psi_0\|\bigr),
\end{split}			
\end{equation*}
and by hypothesis the last two norms are finite. This argument is readily generalized
to provide the following result

\begin{proposition}\label{prop:apriori2}
Suppose that $\psi(t)$ is defined by \eqref{wave1} with Cauchy data $(\psi_0,\dot{\psi}_0)$.
\begin{enumerate}	\aitem
\begin{subequations}\label{estimates1}
	\item Assume that  $(\psi_0,\dot\psi_0)\in H^{n+2k+2, n+2k}_{\cM,m^2}(\cG)$
			$k$, $n\in\N_0$. Then for all $t_1$, $t_2$
			\begin{equation*}	\label{estimatesc}
			\begin{split}
				\|(\partial^{2k}_t\psi)(t_1)-&(\partial^{2k}_t\psi)(t_2)\|_{\cM,n}\\
					&\le |t_1-t_2|\,\max(\rho_\cM(t_1),\rho_\cM(t_2))\\
					&\hspace{3em}\times\bigl(\|\Delta^k_\cM \dot\psi_0\|_{\cM,n}
						+ \max(|t_1|,|t_2|)\,\|\Delta^{k+1}_\cM \psi_0\|_{\cM,n}\bigr)
			\end{split}
			\end{equation*}
		  holds true.
	\item Assume that  $(\psi_0,\dot\psi_0)\in H^{n+2k+2, n+2k+2}_{\cM,m^2}(\cG)$
			$k$, $n\in\N_0$, then for all $t_1$, $t_2$
			\begin{equation*}	\label{estimatesd}
			\begin{split}
				\|(\partial^{2k+1}_t\psi)(t_1)-&(\partial^{2k+1}_t\psi)(t_2)\|_{\cM,n}\\
					&\le |t_1-t_2|\,\max(\rho_\cM(t_1),\rho_\cM(t_2))\\
					&\hspace{3em}\times\bigl(\|\Delta^{k+1}_\cM \psi_0\|_{\cM,n}
						+ \max(|t_1|,|t_2|)\,\|\Delta^{k+1}_\cM \dot\psi_0\|_{\cM,n}\bigr)
			\end{split}
			\end{equation*}
		  holds.
\end{subequations}
\end{enumerate}
\end{proposition}

Next we consider the case where the boundary conditions defined by $\cM$ are such that
$-\Delta_\cM$ is non-negative. Recall that in this case we make the choice $m^2=0$,
and that $\rho_\cM(t)$ is equal to $1$ for all $t\in\R$. Moreover, we remark that
then $\sqrt{-\Delta_\cM}$ is a well-defined self-adjoint operator with domain
$H^1_{\cM}(\cG)$. Let $\psi$ be defined as in~\eqref{wave1}, and let $k\in\N_0$.
Then we get from~\eqref{wavekernelg} for all $t$
\begin{subequations}\label{psipos}
\begin{align}
	\partial_t^{2k} \psi(t)
		&= \cos\bigl(\sqrt{-\Delta_\cM} t\bigr)\,\Delta_\cM^k\psi_0
				+ \sin\bigl(\sqrt{-\Delta_\cM} t\bigr)\,\Delta_\cM^{(2k-1)/2}\dot\psi_0
					\label{psiposa}\\
	\partial_t^{2k+1} \psi(t)
		&= \sin\bigl(\sqrt{-\Delta_\cM} t\bigr)\,\Delta_\cM^{(2k+1)/2}\psi_0
		 + \cos\bigl(\sqrt{-\Delta_\cM} t\bigr)\,\Delta_\cM^k\dot\psi_0,
					\label{psiposb}
\end{align}
\end{subequations}
for Cauchy data $\psi_0$, $\dot\psi_0$ in Sobolev spaces of sufficiently
high degree. Hence we find the following result.

\begin{corollary}	\label{cor:psidiff}
Suppose that the boundary conditions defined by $\cM$ are such that $\Omega_\cM\ge 0$, and hence
$-\Delta_\cM$ is non-negative. Assume furthermore that $\psi$ is defined by~\eqref{wave1}
with Cauchy data $(\psi_0,\dot\psi_0)$. If $(\psi_0,\dot\psi_0)\in H^{n,n-1}_\cM(\cG)$,
$n\in\N$, then $\psi(t)$ is $n$ times strongly continuously differentiable
in $t$.
\end{corollary}

The formulae~\eqref{psipos} lead to alternative estimates as compared to those which we obtain
directly from Proposition~\ref{prop:apriori} for $m^2=0$. They are given in
the next proposition, where we also combine them with the estimates of
Proposition~\ref{prop:apriori}.

\begin{corollary}	\label{cor:psidiffnorm}
Suppose that $\cM$ and $\psi$ are as in the hypothesis of Corollary~\ref{cor:psidiff}.
\begin{subequations}
\begin{enumerate}\aitem
	\item For all $\psi_0$, $\dot\psi_0)\in L^2(\cG)$, the following \emph{a priori}
		  estimate is valid
			\begin{equation*}	
				\|\psi(t)\| \le \|\psi_0\| + |t|\,\|\dot\psi_0\|.
			\end{equation*}
		  If $(\psi_0,\dot\psi_0)\in H^{n+k,n+k-1}$, $n$, $k\in\N_0$, with $n+k\ge 1$,
		  then
			\begin{equation*} 
				\|\partial_t^k \psi(t)\|_{\cM,n}
					\le \|\psi_0\|_{\cM,n+k} + \|\dot\psi_0\|_{\cM,n+k-1}
			\end{equation*}	
		  holds true for all $t$.

	\item Assume that $(\psi_0,\dot\psi_0)\in H^{n+2l,n+2l}_{\cM}$, $l\in\N$, $n\in\N_0$, then
			\begin{equation*}	
				\|\partial^{2l}_t \psi(t)\|_{\cM,n}
					\le \|\psi_0\|_{\cM, n+2l}
					  + \min\Bigl(|t|\|\dot\psi_0\|_{\cM,n+2l},\|\dot\psi_0\|_{\cM,n+2l-1}\Bigr)
			\end{equation*}
			is valid for all $t$. If $(\psi_0,\dot\psi_0)\in H^{n+2l+2,n+2l}_{\cM}$,
			$l\in\N_0$, $n\in\N_0$, then
			\begin{equation*}	
				\|\partial^{2l+1}_t \psi(t)\|_{\cM,n}
					\le \|\dot\psi_0\|_{\cM, n+2l}
					  + \min\Bigl(|t|\|\psi_0\|_{\cM,n+2l+2},\|\psi_0\|_{\cM,n+2l+1}\Bigr)
			\end{equation*}
			holds for all $t$.
\end{enumerate}
\end{subequations}
\end{corollary}

For the analogue of Corollary~\ref{cor:psidiffnorm} in the case that $-\Delta_\cM\ge 0$
we only give the form of the estimates as based on the equations~\eqref{psipos}:

\begin{corollary}	\label{cor:psidiffnormb}
Suppose that $\cM$ and $\psi$ are as in the hypothesis of Corollary~\ref{cor:psidiff}.
Assume that $(\psi_0,\dot\psi_0)\in H^{n+k+1,n+k}_\cM(\cG)$, $k$, $n\in\N_0$. Then
\begin{equation*}	\label{diffpos}
	\|(\partial_t^k\psi)(t_1)-(\partial_t^k\psi)(t_2)\|_{\cM,n}
		\le |t_1-t_2|\,\bigl(\|\psi_0\|_{\cM,n+k+1}+\|\dot\psi_0\|_{\cM,n+k}\bigr)
\end{equation*}
holds true for all $t_1$, $t_2$.
\end{corollary}
Our discussion so far may not be specifically restricted
to the context of metric graphs and self-adjoint
Laplacians defined there. We could instead have considered
any manifold with a self-adjoint Laplacian
$\Delta$ there, for which $-\Delta$ is bounded below and which therefore defines a wave operator.
We would have obtained the same type of estimates.
From now on, however, the specific one-dimensional situation enters.
We continue to consider the case where $\cM$ is such that $-\Delta_\cM\ge 0$.
Let
$f^{(j)}$, $j\in\N$ denote the $j$--th spatial derivative of any function $f$ on $\cG$
for which this derivative exists (in the sense of the derivative of a function or in the
$L^2$--sense). One easily verifies that $\psi^{(2n)}=\Delta_\cM^n \psi$
holds such that relation \eqref{normprimeprime} extends to
$\|\psi^{(2n)}\|=\|\Delta_\cM^n \psi\| $ for all $\psi\in\cD(\Delta_\cM^n)$, while \eqref{psiprimebound1}
extends to $\|\psi^{(2n+1)}\|\le\|(\sqrt{-\Delta_\cM})^{2n+1} \psi\| $
for all $\psi\in\cD((\sqrt{-\Delta_\cM})^{2n+1})$.
Similarly Corollary~\ref{cor:psiprime} provides the following
\begin{corollary}\label{cor:psiprimet}
Suppose that $\cM$ and $\psi$ are as in the hypothesis of Corollary~\ref{cor:psidiff}.
\begin{enumerate}\aitem
\begin{subequations}
	\item If $(\psi_0,\dot\psi_0) \in H^{n+j+k,n+j+k-1}_\cM(\cG)$, $j$, $k$, $n\in\N_0$,
			with $n+k+j\ge 1$, then
			\begin{equation*}	\label{psiprimebounda}
    				\|(\partial_t^k\psi(t))^{(j)}\|_{\cM,n}
        				\le \|\psi_0\|_{\cM,n+j+k} + \|\dot{\psi}_0\|_{\cM,n+j+k-1}
			\end{equation*}
			is valid for all $t$.
	\item If $(\psi_0,\dot\psi_0)\in  H^{n+j+k+1,n+j+k}_\cM(\cG)$, $j$, $k$, $n\in\N_0$,
			then
			\begin{equation*}	\label{psiprimeboundb}
			\begin{split}
				\bigl\|\bigl((\partial_t^k \psi)(t_1)^{(j)}
					&-(\partial_t^k \psi)(t_2)^{(j)}\bigr)\bigr\|_{\cM,n}\\
					&\le |t_1-t_2|\,\bigl( \|\psi_0\|_{\cM,n+j+k+1}
						+ \|\dot\psi_0\|_{\cM,n+j+k}\bigr)
			\end{split}
			\end{equation*}
		  holds true for all $t$.
\end{subequations}
\end{enumerate}
\end{corollary}

We return to the general case, i.e., we do \emph{not} assume that $\cM$ is such
that $-\Delta_\cM$ is non-negative except where otherwise stated.

To establish uniqueness of the solution \eqref{wave1} for given Cauchy data, we
introduce the energy functional.  For any solution $\varphi(t)$ of the wave equation
with $t$ in a time interval $[-T,T]$, say, set
\begin{equation*}\label{uni0}
\begin{split}
    E_\cM(\varphi(t))
        &=\frac{1}{2}\,\|\partial_t\varphi(t)\|_\cG^2
            +\frac{1}{2}\,\langle \varphi(t),-\Delta_\cM\varphi(t)\rangle\\
        &=\frac{1}{2}\|\partial_t\varphi(t)\|_\cG^2
			+\frac{1}{2}\,\|\sqrt{-\Delta_\cM+m^2}\varphi(t)\|_\cG^2
                -\frac{m^2}{2}\,\|\varphi(t)\|_\cG^2
\end{split}
\end{equation*}
which is finite provided $\varphi(t)\in H^1_{\cM,m^2}(\cG)$ and
$\varphi(t)$ is strongly differentiable in $t$ for all $t\in
[-T,T]$. $\langle\varphi(t),-\Delta_\cM\varphi(t)\rangle$ is understood in the sense of
quadratic forms. The factor $1/2$ is inserted in order to conform with the standard
normalization convention.

In particular $E_\cM(\varphi(t))$ is finite for all $t$ when
$\varphi(t)=\psi(t)$ with $\psi(t)$ as given by~\eqref{wave1} with Cauchy data
$(\psi_0,\dot{\psi}_0)\in H^{1,0}_{\cM,m^2}(\cG)$.

\begin{proposition}\label{prop:tconstant}
Let $\varphi$ be any solution of the wave equation~\eqref{wave3} in the time
interval $[-T,T]$ and having the following properties. For all $t\in[-T,T]$
\begin{enumerate} \renewcommand{\labelenumi}{(\alph{enumi})}
    \item{$\varphi(t)\in H^2_{\cM,m^2}(\cG)$,}
    \item{$\varphi(t)$ is three times strongly differentiable in $t$,}
    \item{$\partial_t\varphi(t)\in  H^2_{\cM,m^2}(\cG)$,}
    \item{$\partial_t\varphi(t)$ also satisfies the wave equation.}
\end{enumerate}
Then the energy functional $E_\cM(\varphi(t))$, $t\in[-T,T]$, is time independent.
In addition, if $\cM$ is such that $-\Delta_\cM\ge 0$ holds, then the energy functional
$E_\cM(\varphi(t))$, $t\in[-T,T]$, is non-negative and vanishes if and only if both
$-\Delta_\cM\varphi(t)$ and $\partial_t\varphi(t)$ vanish for all times $t\in[-T,T]$.
\end{proposition}

Again observe that for $\varphi(t)=\psi(t)$ of the form \eqref{wave1} with
Cauchy data $(\psi_0,\dot{\psi}_0)\in H^{3,2}_{\cM,m^2}(\cG)$ the assumptions of
Proposition~\ref{prop:tconstant} are satisfied.
\begin{proof}
By the assumptions we are free to differentiate $E_\cM(\varphi(t))$ with respect to the
time $t\in\R$. We claim that the relation
\begin{equation}\label{commute}
\partial_t\bigl(-\Delta_\cM\varphi(t)\bigr)=-\Delta_\cM\partial_t\varphi(t)
\end{equation}
is valid. In fact, by assumption
$$
\partial_t\bigl(-\Delta_\cM\varphi(t)\bigr)=-\partial_t\bigl(\partial_t^2\varphi(t)\bigr)=
-\partial_t^3\varphi(t)=-\partial_t^2\bigl(\partial_t\varphi(t)\bigr)=
-\Delta_\cM\partial_t\varphi(t).
$$
holds. Another way to obtain this is to observe that $-\Delta_\cM$ is a
linear operator and therefore \eqref{commute} holds.
Thus by standard
calculations the time derivative of $E_\cM(\varphi(t))$ vanishes, thus establishing the
first claim. As for the second claim, assume that $E_\cM(\varphi(t))=0$ for all $t$.
But this implies $\partial_t\varphi(t)=0$ and $\sqrt{-\Delta_\cM}\varphi(t)=0$, which
in turn gives $-\Delta_\cM\psi(t)=0$. The converse is trivial.
\end{proof}

\begin{theorem}\label{thmu}
Let $\cM$ be such that $-\Delta_\cM\ge 0$ and such that $0$ is not an eigenvalue of
$-\Delta_\cM$. Let $\varphi_1(t)$ and $\varphi_2(t)$ be two solutions of the wave equation
for $t\in[-T,T]$ satisfying the assumptions of Proposition~\ref{prop:tconstant} and
with the same initial values,
$$
    \varphi_1(t=0)=\varphi_2(t=0),\qquad \partial_t\varphi_1(t=0)=\partial_t\varphi_2(t=0).
$$
Then $\varphi_1(t)=\varphi_2(t)$ holds for all $t\in[-T,T]$. In particular $\psi$ as
given by \eqref{wave1} is the unique solution of the wave equation for given Cauchy
data $(\psi_0,\dot{\psi}_0)\in H^{3,2}_{\cM,m^2}(\cG)$.
\end{theorem}

For a given metric graph, necessary and sufficient conditions on $\cM$ for
$-\Delta_\cM$ to have~$0$ as an eigenvalue are given in \cite{Schrader}, see also
Corollary~\ref{cor:qpos}. If $\psi^0$ is such an eigenfunction, it has to be
constant on each edge and in particular zero on each external edge. Also $\psi^0(t)$ as given by
\eqref{wave1} with Cauchy data $(\psi^0,\dot{\psi}=0)$ satisfies $\psi^0(t)=\psi^0$
for all~$t\in\R$.

\begin{proof}[Proof of Theorem~\ref{thmu}]
Standard and well known arguments can now be used. Indeed, $\varphi_1(t)-\varphi_2(t)$ is
a solution of the wave equation with vanishing initial data and we can use the
previous proposition.
\end{proof}

In order to establish finite propagation speed, we introduce a local form of the
energy functional. As a motivation we use \eqref{sesquilinear:2} to rewrite the
energy functional as
\begin{equation}\label{uni1}
E_\cM(\psi(t))=\frac{1}{2}\,\left(\|\partial_t\psi(t)\|_\cG^2
				+\|\psi(t)^\prime\|_\cG^2
				+\langle[\psi(t)], \Omega_\cM[\psi(t)]\rangle_{{}^d\cK}\right),
\end{equation}
which is finite for all $t$ provided the Cauchy data $(\psi_0,\dot{\psi}_0)$ are such
that $\psi_0\in H^2_{\cM,m^2}(\cG)$, $\dot{\psi}_0\in  H^1_{\cM,m^2}(\cG)$,
cf.\ the remarks after Proposition~\ref{prop:sesq} and Lemma~\ref{lem_wkernel}.

The first two terms on the right hand side form the energy functional for solutions of the
wave equation on smooth manifolds, see e.g.~\cite{Cheeger_Gromov_Taylor, Evans, Taylor, TaylorI}. So it is the last term which is special for the present context of metric graphs, which are singular manifolds. Of these three terms it is the only one, in which the boundary condition $\cM$
enters and, as we shall see, in a manageable way. For the remainder of this
section we assume that $\cM$ is such that $\Omega_\cM\ge 0$, and hence also
$-\Delta_\cM\ge 0$, as well as $(\psi_0,\dot\psi_0)\in H^{4,3}_\cM(\cG)$. Then
Lemma~\ref{lem:Schwarz} in the appendix entails that $\|\psi(t)^\prime\|^2$ is
differentiable in $t$. Since $\langle \psi(t),-\Delta_\cM\psi(t)\rangle
= \langle \psi(t),-\partial_t^2\psi(t)\rangle $ is also differentiable in $t$ (see
Proposition~\ref{prop:waveeq}), we conclude that
$$
\langle[\psi(t)], \Omega_\cM[\psi(t)]\rangle_{{}^d\cK}
	=\langle \psi(t),-\Delta_\cM\psi(t)\rangle-\|\psi(t)^\prime\|_\cG^2
$$
is differentiable with respect to~$t$ too.

Actually, more is valid and will be used, namely we also have differentiability of
the boundary values $[\psi(t)]$ themselves.

\begin{lemma}\label{lem:diffpsi} Assume the boundary condition $\cM$ is such that
$\Omega_\cM\ge0$ and hence also $-\Delta_\cM\ge 0$ is valid. For Cauchy data
$(\psi_0,\dot\psi_0)$ in $H^{4,3}_{\cM}(\cG)$ the boundary value $[\psi(t)]$ is
continuously differentiable in $t$ and
\begin{equation}\label{partial1}
	\partial_t[\psi(t)]=[\partial_t\psi(t)]
\end{equation}
holds.
\end{lemma}
Observe in this context that since $\psi(t)\in {\rm Dom}(-\Delta_\cM)$ is valid for all $t$,
the relation $\cP_\cM[\psi(t)]=[\psi(t)]$ holds for all $t$ which upon differentiation gives
\begin{equation}\label{partial2}
\cP_\cM[\partial_t\psi(t)]=[\partial_t\psi(t)].
\end{equation}
Alternatively \eqref{partial2} follows from $\partial_t\psi(t)\in {\rm Dom}(-\Delta_\cM)$, which in turn
follows from the assumptions on the Cauchy data. The proof of Lemma~\ref{lem:diffpsi} is
based on Sobolev estimates in conjunction with Corollary~\ref{cor:psiprimet} and
will be given in the appendix.


\section{Finite Propagation Speed}\label{sec:finprop}


The form \eqref{uni1} allows us to introduce a local energy functional. Fix a point
$p\in\cG$ and a time $t_0\in\R$. For any $\psi$ of the form \eqref{wave1} with Cauchy data
$(\psi_0,\dot{\psi}_0)\in H^{2,1}_{\cM,m^2}(\cG)$ and $0\le t\le t_0$,
the time dependent \emph{local energy functional} is defined as
\begin{equation}\label{energy}
    e(t)=\frac{1}{2}\left(\|\partial_t\psi(t)\|^2_{B(p,t_0-t)}
            +\|\psi(t)^\prime\|^2_{B(p,t_0-t)}
            +\langle[\psi(t)], \Omega_t[\psi(t)]\rangle_{{}^d\cK}\right),
\end{equation}
where
\begin{equation*}\label{omegat}
    \Omega_t=\Omega_{\cM,\,B(p,t_0-t)\cap \cV}
            =P_t\Omega_\cM P_t=P_t\Omega_\cM=\Omega_\cM P_t,
\end{equation*}
cf.~\eqref{qv}, with
\begin{equation*}\label{pt}
    P_t=\sum_{v\in B(p,t_0-t)\cap \cV}P_v
\end{equation*}
satisfying $P_t\le P_{t^\prime}$ for $t^\prime\le t$. $\Omega_t$ is piecewise
constant in $t$ with possible jumps at $\cT(p)$. Observe that all terms on the
right hand side of~\eqref{energy} are finite: By Proposition~\ref{prop:waveeq},
the hypothesis $(\psi_0,\dot\psi_0)\in H^{2,1}_{\cM,m^2}(\cG)$ implies that
$\psi(t)\in H^2_{\cM,m^2}(\cG) = {\rm Dom}(\Delta_\cM)$ for all $t\in\R$. Thus, on every
edge of $\cG$, $\psi(t)$ and $\psi(t)'$ are continuous functions, and in particular
their boundary values $[\psi(t)]$ at the vertices of $\cG$ are well-defined and finite.

The initial value of $e(t)$ can be expressed in terms of the Cauchy data themselves as
\begin{equation*}\label{energy0}
    e(t=0)=\frac{1}{2}\left(\|\dot{\psi}_0\|^2_{B(p,t_0)}+
            \|\psi_0^\prime\|^2_{B(p,t_0)}+\langle[\psi_0],
            \Omega_{t=0}[\psi_0]\rangle_{{}^d\cK}\right).
\end{equation*}
We will be interested in the situation when the condition $\Omega_\cM\ge 0$ is
satisfied and then obviously $0\le e(t)$ for all $t$. Also by Lemma \ref{lem:qs}
\begin{equation}\label{omegamon}
    \Omega_t\le \Omega_{t^\prime},\quad t^\prime\le t,
\end{equation}
is valid. Now, for $e(t)$ to vanish when $0\le \Omega_\cM$, it is necessary that
both $\partial_t\psi(t)$ and $\psi(t)^\prime$ vanish on $B(p,t_0-t)$. In particular
$\psi(t)$ is then piecewise constant, that is $\psi_i(t)$ is constant on each
$B(p,t_0-t)\cap I_i,\;i\in\cE\cup\cI$, which is a connected set. In the case where actually
$\Omega_\cM> 0$ holds, for $e(t)$ to vanish it is necessary and sufficient that both
$\partial_t\psi(t)$ and $\psi(t)^\prime$ vanish on $B(p,t_0-t)$ and that
$P_t[\psi(t)]=0$.

We want to show that $e(t)$ is non-increasing in $t$. To establish
this we need a couple of Lemmas. The first one is a local version of
Proposition~\ref{prop:sesq}. For its formulation we need an adaption of the familiar notion of a
normal derivative to the present context.

\begin{definition}\label{def:2}
Assume $0<t_0-t\notin\cT(p)$ with $p\cong (k,y)$. The \emph{inward normal
derivative} of $\psi$ at $q\in\partial B(p,t_0-t)$ with coordinate $q\cong (i,x)
\;(0<x<a_i,\,i\in\cE\cup\cI)$ is defined as
\begin{equation*}
    \partial_{\bf n}\psi(q)
        =\begin{cases}
            \phantom{-}\psi_i^\prime(x),
                & \text{if $k=i$, $x<y$, or if $i\neq k$ and $[x,a_i]\subset B(p,t_0-t)$,}\\[1ex]
            -\psi_i^\prime(x),
                & \text{if $k=i$, $y<x$, or if $i\neq k$ and $[0,x]\subset B(p,t_0-t)$.}
\end{cases}
\end{equation*}
\end{definition}

The sign convention is made to conform with the sign convention in the definition
\eqref{lin1:add} of $\underline{\psi}^\prime$ and hence of $[\psi]$. As an example
consider the case $k=i\in \cE$, again with $ p=(k,y)$ and in addition $t$ so close
to $t_0$ that $0<t_0-t<y$. Then $B(p,t_0-t)$ is an interval on $I_k\cong [0,\infty)$
of the form $[y-t_0+t,y+t_0-t]$ centered at $y$ and of length
$|B(p,t_0-t)|=2(t_0-t)$. So $\partial B(p,t_0-t)$ consists of the two points
$(k,y-t_0+t)$ and $(k,y+t_0-t)$ such that $\partial_{\bf
n}\psi(k,y-t_0+t)=\psi_k^\prime(y-t_0+t)$ and $\partial_{\bf
n}\psi(k,y+t_0-t)=-\psi_k^\prime(y+t_0-t)$.

\begin{lemma}\label{lem:2}
For every boundary condition $\cM$, and any $t_0-t\in\R_+\setminus\cT(p)$
the relation
\begin{equation}\label{partialint}
    \langle\varphi, -\Delta_\cM \psi\rangle_{B(p,t_0-t)}
        = \langle\varphi^\prime,\psi^\prime\rangle_{B(p,t_0-t)}
            +\langle[\varphi], \Omega_{t}[\psi]\rangle_{{}^d\cK}+
            \sum_{q\in \partial B(p,t_0-t)}\overline{\varphi(q)}\partial_{\bf n}\psi(q)
\end{equation}
is valid for any $\varphi,\psi\in {\rm Dom}(-\Delta_\cM)$.
\end{lemma}

\begin{proof}
Observe that by the remark after Proposition~\ref{prop:sesq} or --- in the case that
$-\Delta_\cM\ge 0$ --- more easily by Corollary~\ref{cor:psiprime} both
$\varphi^\prime$ and $\psi^\prime$ are elements in
$L^2(\cG)$. Furthermore since ${\rm Dom}(-\Delta_\cM)\subset \cD$ all terms on the
right hand side of \eqref{partialint} are well defined and finite. This allows us to
perform an integration by parts and to use Green's identity. Firstly there are boundary
contributions at those vertices, which are contained in $B(p,t_0-t)$, and secondly
at points of the boundary $\partial B(p,t_0-t)$ giving
\begin{equation*}
\begin{split}
    \langle\varphi, &-\Delta_\cM \psi\rangle_{B(p,t_0-t)}\\
        &= \langle\varphi^\prime,\psi^\prime\rangle_{B(p,t_0-t)}
\            +\sum_{v \in B(p,t_0-t)} \langle ^d Q_v[\varphi], \Omega\, ^dQ_v[\psi]\rangle_{{}^d\cK}
            + \sum_{q\in \partial B(p,t_0-t)}\overline{\varphi(q)}\partial_{\bf n}\psi(q).
\end{split}
\end{equation*}
Now we insert $P_\cM[\varphi]=[\varphi]$ and $P_\cM[\psi]=[\psi]$, valid due to the
assumption $\varphi,\psi\in {\rm Dom}(-\Delta_\cM)$, into the second term. Using in
addition \eqref{pmv} we obtain
\begin{align*}
    \sum_{v \in B(p,t_0-t)}\langle\; ^dQ_v[\varphi], \Omega\; ^dQ_v[\psi]\rangle_{{}^d\cK}
        &=\sum_{v \in B(p,t_0-t)}\langle\;[\varphi],P_\cM \;^d Q_v\, \Omega\, ^dQ_v\;
                                                        P_\cM \,[\psi]\rangle_{{}^d\cK}\\[1ex]
        &=\langle\;[\varphi],\Omega_t\,[\psi]\rangle_{{}^d\cK}.\qedhere
\end{align*}
\end{proof}

\begin{proposition}\label{prop:ediff}
Assume that the boundary condition $\cM$ is such that $\Omega_\cM\ge0$ and hence
$-\Delta_\cM\ge 0$. Also let the Cauchy data $(\psi_0,\dot{\psi}_0)$ be such that
$\psi_0\in H^4_\cM(\cG)$ and $\dot{\psi}_0\in H^3_\cM(\cG)$. Then $e(t)$ is
differentiable at all points $t$ with $t_0-t\in\R_+\setminus \cT(p)$ and satisfies
$\partial_t e(t)\le 0$ there.
\end{proposition}

\begin{proof}
We differentiate $e(t)$ under the assumption on $t$ that $\partial B(p,t_0-t)\cap
\cV=\emptyset$ which in particular means that $\Omega_s$ is constant for all $s$
close to $t$. We use \eqref{partial1} and obtain
\begin{equation*}    \label{energy1}
\begin{split}
    \partial_t  e(t)
        &=\frac{1}{2}\langle\partial_t^2\psi(t),\partial_t\psi(t)\rangle_{B(p,t_0-t)}
            +\frac{1}{2}\langle\partial_t\psi(t),\partial_t^2\psi(t)\rangle_{B(p,t_0-t)}\\
        &\hspace{2em} + \frac{1}{2}\langle\partial_t\psi(t)^\prime,\psi(t)^\prime\rangle_{B(p,t_0-t)}
                        +\frac{1}{2}\langle\psi(t)^\prime,\partial_t\psi(t)^\prime\rangle_{B(p,t_0-t)}\\
        &\hspace{2em} - \frac{1}{2}\sum_{q\in \partial B(p,t_0-t)}\left(\left|\partial_t\psi(t,q)\right|^2
                        +\left|\psi^\prime(t,q)\right|^2\right)\\
        &\hspace{2em} + \frac{1}{2}\langle[\partial_t\psi(t)],\Omega_{t}[\psi(t)]\rangle_{{}^d\cK}
                        +\frac{1}{2}\langle[\psi(t)],\Omega_{t}[\partial_t\psi(t)]\rangle_{{}^d\cK}
\end{split}
\end{equation*}
with the abbreviation $\psi(t,q)=\psi(t)(q)$.
In a next step we first invoke the
wave equation~\eqref{wave3} for the first two terms on the right hand side, and then use
Lemma~\ref{lem:2}. This gives
\begin{equation}    \label{energy2}
\begin{split}
    \partial_t e(t)
        &= -\frac{1}{2}\sum_{q\in \partial B(p,t_0-t)}\Big(\left|\partial_t\psi(t,q)\right|^2
              +\left|\psi(t,q)^\prime\right|^2\\
        &\hspace{8em} +\overline{\partial_t\psi(t,q)}\,\partial_{\bf n}\psi(t,q)
                      +\overline{\partial_{\bf n}\psi(t,q)}\,\partial_t\psi(t,q)\Big)\\
        &= -\frac{1}{2}\sum_{q\in \partial B(p,t_0-t)}\left|\partial_t\psi(t,q)
              +\partial_{\bf n}\psi(t,q)\right|^2\le 0,
\end{split}
\end{equation}
and the proof is finished.
\end{proof}
Next we look at what happens when $t_0-t\in\cT(p)$. As a motivation for our further
procedure, we show why relation \eqref{energy2} fails when $\partial B(p,t_0-t)$
contains coinciding points. To simplify the discussion we assume there is only one
coinciding point $q\in\partial B(p,t_0-t)$ and that $\partial
B(p,t_0-t)\cap\cV=\emptyset$. With the notation used in Definition~\ref{def:coinc},
for all $s>t$ sufficiently close to $t$ there will be a contribution to $\partial_s
e(s)$ of the form
\begin{equation}\label{coinclim}
\begin{split}
    -\frac{1}{2}&\overline{\partial_t\psi(s,q_l(t_0-s))}\;\partial_{\bf n}\psi(s,q_l(t_0-s))
    -\frac{1}{2}\overline{\partial_t\psi(s,q_r(t_0-s))}\;\partial_{\bf n}\psi(s,q_r(s))\\
    &-\frac{1}{2}\overline{\partial_{\bf n}\psi(s,q_l(t_0-s))}\;\partial_t\psi(s,q_l(t_0-s))
    -\frac{1}{2}\overline{\partial_{\bf n}\psi(s,q_r(t_0-s))}\;\partial_t\psi(s,q_r(t_0-s)).
\end{split}
\end{equation}
Since $\lim_{s\downarrow t}q_l(t_0-s)=\lim_{s\downarrow t}q_r(t_0-s)=q$, by continuity
$$
\lim_{s\downarrow t}\partial_t\psi(s,q_l(t_0-s))=\lim_{s\downarrow t}\partial_t\psi(s,q_r(t_0-s))
$$
while
$$
\lim_{s\downarrow t}\partial_{\bf n}\psi(q_l(s),s)
=-\lim_{s\downarrow t}\partial_{\bf n}\psi(q_r(s),s)
$$
so the terms in \eqref{coinclim} cancel pairwise when $s$ decreases to $t$.

\begin{proposition}\label{prop:discon}
Assume the boundary condition $\cM$ is such that $\Omega_\cM\ge0$ and hence also
$-\Delta_\cM\ge 0$. Also let the Cauchy data $(\psi_0,\dot{\psi}_0)$ be such that
$\psi_0\in H^4_\cM(\cG)$ and $\dot{\psi}_0\in H^3_\cM(\cG)$. The following relation
holds
\begin{equation*}\label{discon}
\lim_{s\uparrow t}e(s)=e(t)\ge\lim_{s\downarrow t}e(s)
\end{equation*}
for all $t_0-t\in\cT(p)$.
\end{proposition}
For the proof we need the
\begin{lemma}\label{lem:cont}
Under the assumptions of Proposition~\ref{prop:discon} on $\cM$ and the Cauchy data, the map
\begin{equation*}\label{cont}
t\quad \mapsto\quad \frac{1}{2}\left(||\partial_t\psi(t)||^2_{B(p,t_0-t)}
+||\psi(t)^\prime||^2_{B(p,t_0-t)}\right)
\end{equation*}
is continuous in $t\in\R_+$.
\end{lemma}
The proof of this lemma will be given in the appendix and is based on
Lemma~\ref{app:lem:cont1},
whose proof is also given there.

Under the assumption $\Omega_\cM\ge 0$ and by \eqref{omegamon}
\begin{equation}\label{omegacont}
\lim_{s\uparrow t}\Omega_s\ge \lim_{s\downarrow t}\Omega_s.
\end{equation}
More explicitly
$$
\lim_{s\uparrow t}\Omega_s=\Omega_t=\sum_{v\in\partial B(t_0-t)\cap\cV}\Omega_v
+\lim_{s\downarrow t}\Omega_s\ge \lim_{s\downarrow t}\Omega_s.
$$
We combine \eqref{omegacont} with Lemma~\ref{lem:diffpsi} and conclude
$$
\lim_{s\uparrow t}\langle[\psi(s)],\Omega_{s}[\psi(s)]\rangle_{{}^d\cK}=
\langle[\psi(t)],\Omega_{t}[\psi(t)]\rangle_{{}^d\cK}\ge
\lim_{s\downarrow t}\langle[\psi(s)],\Omega_{s}[\psi(s)]\rangle_{{}^d\cK}.
$$
This result combined with Lemma \ref{lem:cont} concludes the proof of Proposition \ref{prop:discon}.
In turn the Propositions~\ref{prop:ediff} and~\ref{prop:discon} give the first part of

\begin{theorem}[Finite propagation speed]
Assume the boundary condition $\cM$ is such that $\Omega_\cM\ge 0$ and hence $-\Delta_\cM\ge 0$.
For Cauchy data $\psi_0\in H^4_\cM(\cG)$ and $\dot{\psi}_0\in H^3_\cM(\cG)$, let
$\psi(t)$ be defined by \eqref{wave1}. Fix a point $p$ and a time $t_0>0$.
Then $e(t)$ as defined by \eqref{energy} is non-negative and non-increasing for
$0\le t\le t_0$.
If $\psi_0$ and $\dot{\psi}_0$ both vanish on $B(p,t_0)$, then $\psi(t,q)$ vanishes on the cone
$$
\cC(p,t_0)=\left\{\;(t,q)\;|\;0\le t\le t_0,\; d(q,p)\le t_0-t\;\right\}\subset \cG\times \R_+
$$
with vertex at $(p,t_0)$.
\end{theorem}
\begin{proof}
The proof of the second part now uses standard arguments, see e.g.~\cite{Evans}. By
assumption $e(t=0)=0$. Hence by the first part of the theorem, $e(t)=0$ for all $0\le
t\le t_0$. Thus $\partial_t\psi(t,q)=\psi^\prime(t,q)=0$ for $(t,q)\in \cC(p,t_0)$. As a
consequence
$$
\psi(t,q)=\psi_0(q)+\int_{0}^t\partial_s\psi(s,q)ds=0
$$
for $(q,t)\in \cC(p,t_0)$.
\end{proof}

\begin{appendix}
\section{Proof of Lemmas  \ref{lem:diffpsi} and \ref{lem:cont}}
\label{app:psitcont}
\renewcommand{\theequation}{\mbox{\Alph{section}.\arabic{equation}}}
\setcounter{equation}{0}
Throughout the appendix $\cM$ is chosen to be such that $\Omega_\cM\ge 0$ holds.
We recall the Sobolev inequality in 1 dimension, see e.g.~\cite[Theorem~8.5]{Lieb_Loss}.
Any function $f$ in the Sobolev space $H^1(\R)$ is bounded and satisfies the estimate
\begin{equation}\label{sob1}
	\|f\|_\infty^2\le \frac{1}{2}\left(\|f\|_{L^2(\R)}^2+\|f^\prime\|_{L^2(\R)}^2\right).
\end{equation}
In order not to burden the notation, here and in what follows $||\,\cdot\,||_\infty$ will always
denote the $L^\infty$ norm while $||\,\cdot\,||$ is the $L^2$ norm in the respective
context. The Sobolev inequality easily carries over to our context
where $\R$ is replaced by $\cG$
\begin{equation*}\label{sob11}
	\|\psi\|_\infty^2\le \frac{1}{2}\left(\|\psi\|^2+\|\psi^\prime\|^2\right).
\end{equation*}
This inequality follows by simple arguments from \eqref{sob1}, which are omitted here.
More generally, for $\psi\in H^{j+1}_\cM(\cG)$, $j\in\N_0$,
\begin{equation*}\label{sob12}
	\|\psi^{(j)}\|_\infty^2
		\le \frac{1}{2}\left(\|\psi^{(j)}\|^2+\|\psi^{(j+1)}\|^2\right)
\end{equation*}
holds. This inequality is now combined with Corollary~\ref{cor:psiprimet}
to obtain several estimates for $\psi(t)$,
as defined by \eqref{wave3} with Cauchy data $(\psi_0,\dot{\psi}_0)$.
For $n\in\N$,
$(\psi_0,\dot\psi_0)\in H^{n,n-1}_{\cM}(\cG)$ introduce
\begin{equation*}
	A_1(\psi_0,\dot\psi_0,t)
		=\Bigl(\|\psi_0\|^2 + \|\psi_0\|_{\cM,1}^2 + (1+t^2)\,\|\dot\psi_0\|^2\Bigr)^{1/2}
\end{equation*}
and for $n\ge 2$,
\begin{equation*}
	A_n(\psi_0\dot\psi_0)
		= \Bigl(\|\psi_0\|_{\cM,n-1}^2 + \|\psi_0\|_{\cM,n}^2
			+ \|\dot\psi_0\|_{\cM,n-2}^2 + \|\psi_0\|_{\cM,n-1}^2\Bigr)^{1/2}.
\end{equation*}
We leave out the proof of the following lemma.
\begin{lemma}\label{app:lem:sob}
Suppose that $\cM$ is such that $-\Delta_\cM\ge 0$, and that $\psi(t)$ is
given by \eqref{wave1} with Cauchy data $(\psi_0,\dot\psi_0)$. Then the following
estimates hold true for all $t$, $t_1$, $t_2$:
\begin{enumerate}\aitem
\begin{subequations}\label{psisob1}
	\item  for $(\psi_0,\dot\psi_0)\in H^{1,0}_{\cM}(\cG)$,
			\begin{equation}	\label{psiob1a}
				\|\psi(t)\|_\infty \le A_1(\psi_0,\dot\psi_0,t),
			\end{equation}
		  and for $(\psi_0,\dot\psi_0)\in H^{j+k+1,j+k}_\cM(\cG)$,
		  $j$, $k\in\N_0$, with $j+k\ge 1$,
			\begin{equation}	\label{psiob1b}
				\bigl\|(\partial_t^k\psi)(t)^{(j)}\bigr\|_\infty
					\le A_{j+k+1}(\psi_0,\dot\psi_0);
			\end{equation}
	\item for $(\psi_0,\dot\psi_0)\in H^{j+k+2,j+k+1}_{\cM}(\cG)$, $j$, $k\in\N_0$,
			\begin{equation}	\label{psiob1c}
				\bigl\|\bigl((\partial^k_t\psi)(t_1)\bigr)^{(j)}-
					\bigl((\partial^k_t\psi)(t_2)\bigr)^{(j)}\bigr\|_\infty
					\le |t_1-t_2|\,A_{j+k+2}(\psi_0,\dot\psi_0).
			\end{equation}
\end{subequations}
\end{enumerate}
\end{lemma}
In a next step we will prove the continuity of
\begin{equation}\label{mixedcont}
\left(\partial_t\psi(t)\right)^\prime_i(x)=\frac{\partial}{\partial x}
\frac{\partial}{\partial t}\psi_i(t,x)
\end{equation}
in both $t\in\R$ and in $x\in I_i,\,i\in\cE\cup\cI$.
This will enable us to establish both the existence of and the
equality with the other mixed partial second derivative
\begin{equation}\label{mixed2}
\frac{\partial}{\partial t}\frac{\partial}{\partial x}\psi_i(t,x).
\end{equation}
To this end, we assume from now on that the Cauchy data $(\psi_0,\dot\psi_0)$ belong
to $H^{4,3}_\cM(\cG)$. So by Proposition~\ref{prop:waveeq}
$\psi(t)\in H^4_\cM(\cG)$ and  $\partial_t\psi(t)\in H^3_\cM(\cG)$ and therefore both
have spatial derivatives up to third order in $L^2(\cG)$. As a consequence the restrictions of
both to each edge have absolutely continuous spatial derivatives up to
order two. Thus on every edge we may consider $\psi(t)^{(j)}$, $\partial_t\psi(t)^{(j)}$,
$t\in\R$, $j=0$, $1$, $2$, as \emph{bona fide} functions, and
in particular their $L^\infty$ norms equal their $\sup$--norms.

Consider a fixed edge $I_i$ of $\cG$, $x_1$, $x_2\in I_i$, and let $j$, $k=0$,~$1$.
Then the mean value theorem together with inequality~\eqref{psiob1b} gives
\begin{align*}
	\sup_{t\in\R}\,\bigl|(\partial_x^j\partial_t^k\psi)_i(t,x_1)
						-(\partial_x^j\partial_t^k\psi)_i(t,x_2)\bigr|
			&\le |x_1-x_2|\,\sup_{t\in\R}\,\bigl\|\partial_t^k \psi(t)^{(j+1)}\bigr\|_\infty\\
			&\le |x_1-x_2|\,A_{j+k+2}(\psi_0,\dot\psi_0),
\end{align*}
and our assumptions entail that $A_{j+k+2}(\psi_0,\dot\psi_0)$ is finite for all
$j$, $k=0$,~$1$. Hence we have shown

\begin{lemma}\label{app:lem:cont1}
Suppose that $\cM$ is such that $-\Delta_\cM\ge 0$, and that $\psi(t)$ is
given by \eqref{wave1} with Cauchy data $(\psi_0,\dot\psi_0)$ in $H^{4,3}_\cM(\cG)$.
Then both, the family of functions  $\{\psi(t),\,t\in\R\}$ and the family of their
derivatives $\{\psi(t)^\prime,\,t\in\R\}$, are uniformly bounded on $\cG$, and
uniformly equicontinuous on each edge of $\cG$. The same is valid for the family
$\{\partial_t\psi(t),\,t\in\R\}$ and its derivatives
$\{(\partial_t\psi(t))^\prime,\,t\in\R\}$.
\end{lemma}

On the other hand, consider $t_1$, $t_2\in\R$, $j$, $k=0$,~$1$. Then~\eqref{psiob1c}
yields
\begin{align*}
	\sup_{x\in I_i}\, \bigl|\partial_x^j\partial_t^k\psi_i(t_1,x)
						-\partial_x^j\partial_t^k\psi_i(t_2,x)\bigr|
		&\le \bigl\|(\partial_t^k\psi)(t_1)^{(j)}-(\partial_t^k\psi)(t_2)^{(j)}\bigr\|_\infty\\
		&\le |t_1-t_2|\,A_{j+k+2}(\psi_0,\dot\psi_0).
\end{align*}
Hence on every edge $I_i$ and for all $x\in I_i$, the mappings
$t\mapsto \psi_i(t,x)$,
$\partial_x \psi_i(t,x)$, $\partial_t\psi_i(t,x)$, and $\partial_x\partial_t\psi_i(t,x)$ are
uniformly continuous, uniformly in $x\in I_i$. Thus we have established: If the Cauchy
data $(\psi_0,\dot\psi_0)$ belong to $H^{4,3}_\cM(\cG)$, then
for every edge $I_i$ of $\cG$ the maps
\begin{equation*}
    (t,x)\quad \mapsto\quad
        \begin{cases}	\displaystyle
            \psi_i(t,x),\\[2ex]
            	\displaystyle\frac{\partial}{\partial x}\,\psi_i(t,x),\\[2ex]
            	\displaystyle\frac{\partial}{\partial t}\,\psi_i(t,x),\\[2ex]
            	\displaystyle\frac{\partial^2}{\partial x\,\partial t}\,\psi_i(t,x),
        \end{cases}\qquad (t,x)\in  \R\times I_i,
\end{equation*}
are uniformly continuous. So we can apply the lemma of Clairaut--Schwarz, see
e.g.~\cite[Theorem~7.A.11, p.~194]{Voxman_Goetschel} for the version
we use, to conclude

\begin{lemma}\label{lem:Schwarz}
Suppose that $\cM$ is such that $-\Delta_\cM\ge 0$, and that $\psi(t)$ is
given by \eqref{wave1} with Cauchy data $(\psi_0,\dot\psi_0)$ in $H^{4,3}_\cM(\cG)$.
Then for every edge $I_i$ of $\cG$ the mixed partial derivative~\eqref{mixedcont}
exists and equals the other mixed partial derivative \eqref{mixed2}, which is uniformly
continuous for $(t,x)\in \R\times I_i$.
\end{lemma}

Lemma~\ref{lem:diffpsi} is a direct consequence of this lemma as is
Lemma~\ref{lem:cont} in combination with the following observation. The volume
$$
\mu(B(p,t_0-t))=\int_{q\in B(p,t_0-t)}dq
$$
of $B(p,t_0-t)$ is continuous in $t$. More precisely, the uniform estimate $(t_2\le t_1)$
\begin{equation*}
\begin{split}
	0\le \mu\bigl(B(p,t_0-t_2)\setminus B(p,t_0-t_1)\bigr)
		&= \mu\bigl(B(p,t_0-t_2)\bigr)-\mu\bigl(B(p,t_0-t_1)\bigr)\\
		&\le (t_1-t_2)\,2(|\cE|+|\cI|)
\end{split}
\end{equation*}
is easily established.
\end{appendix}

\end{document}